\documentclass[runningheads]{llncs}

\usepackage[T1]{fontenc}

\usepackage{amsmath,amssymb}
\usepackage{hyperref}

\usepackage{graphicx}
\usepackage{xcolor}
\usepackage{subcaption}
\usepackage{forest}
\usepackage{enumitem}
\usepackage{tabularx}

\setlist[itemize]{
  leftmargin=*,
  topsep=2pt,
  itemsep=0pt,
  parsep=0pt,
  partopsep=0pt
}

\setlist[enumerate]{
  leftmargin=*,
  topsep=2pt,
  itemsep=0pt,
  parsep=0pt,
  partopsep=0pt
}

\spnewtheorem{observation}
  {Observation}
  {\bfseries}
  {\itshape}

\begin{document}

\title{When Robots Crash: Optimal Asynchronous Gathering at Weber Meeting Nodes}
\titlerunning{Crash-Tolerant Gathering at Weber Meeting Nodes}

\author{
Animesh Maiti\inst{1}\thanks{Corresponding author.} \and
Prakhar Shukla\inst{1} \and
Abhinav Chakraborty\inst{2} \and
Subhash Bhagat\inst{1}
}

\authorrunning{A. Maiti et al.}

\institute{
Department of Mathematics,\\
Indian Institute of Technology Jodhpur,\\
Jodhpur, Rajasthan, India\\
\email{p22ma201@iitj.ac.in,
       p23ma0011@iitj.ac.in,
       sbhagat@iitj.ac.in}
\and
Department of Mathematics,\\
Birla Institute of Technology Mesra,\\
Ranchi, Jharkhand, India\\
\email{abhinavchakraborty@bitmesra.ac.in}
}

\maketitle
%
\begin{abstract}

We study the \textit{optimal gathering} problem over a finite set of designated \textit{meeting nodes} for \textit{asynchronous, anonymous,} and \textit{oblivious} mobile robots on an infinite grid under crash faults. The robots have global visibility and strong multiplicity detection, but share neither a coordinate system nor chirality. The objective is to gather all non-faulty robots at a \textsc{Weber Meeting Node}, minimizing the total Manhattan distance from their initial positions. Up to $n-2$ robots may crash permanently, and such crashes are indistinguishable from arbitrary delays. Existing approaches often rely on a designated robot to break symmetry, whose crash may block the remaining robots indefinitely. Instead, our approach enables every robot to independently select the same target from its snapshot, while target-dependent restricted shortest paths preserve the target as a \textsc{Weber Meeting Node}. We prove that, under strong multiplicity detection, optimal gathering is impossible from certain fully symmetric configurations. For all remaining configurations, our algorithm \textsc{CrashTolerantWeberGathering()} selects a unique common target, preserves its optimality throughout the execution, and allows non-faulty robots to progress without waiting for crashed robots, thereby guaranteeing gathering in finite time.

\end{abstract}

\section{Introduction}
\textbf{Background and motivation.}
Swarm robotics is inspired by the collective behaviour of social insects such
as ants, bees, and termites. A robot swarm consists of several simple and
inexpensive robots that work together, without centralized control or explicit
communication, to perform tasks such as exploration, pattern formation,
convergence, and gathering. The main question is how robots with limited
computational capabilities can coordinate to achieve a common goal.

In the classical distributed model, robots are
\emph{autonomous} (operate without centralized coordination),
\emph{homogeneous} (identical in capabilities and behaviour),
\emph{anonymous} (have no unique identifiers),
\emph{oblivious} (have no memory of past actions), and
\emph{silent} (do not communicate explicitly). Each robot executes the same
deterministic algorithm and repeatedly performs a
\emph{Look--Compute--Move} (\textsc{LCM}) cycle: it observes the current
configuration, computes a destination, and either remains stationary or moves
toward that destination. The robots are disoriented and share neither a
coordinate system nor chirality. We consider robots deployed on an infinite grid, where initially each robot occupies a distinct grid node and moves along grid edges. Movements are instantaneous, so robots are observed
only at grid nodes. They have \emph{global visibility} and
\emph{strong global multiplicity detection}; hence, they observe the complete
configuration and know the exact number of robots at each occupied node. Every
robot makes its decision only from the snapshot expressed in its local
coordinate system. The grid contains a finite set of distinguished nodes, called
\emph{meeting nodes}, which may represent charging stations, communication
points, maintenance facilities, or evacuation sites. In the
\emph{gathering over meeting nodes} problem, the robots must gather at one of
these designated nodes rather than at an arbitrary grid node. For
\emph{optimal gathering}, the destination must be a
\textsc{Weber Meeting Node}, that is, a meeting node minimizing the total
Manhattan distance from the initial robot positions. The robots operate under a fair \emph{asynchronous} (\textsc{ASYNC})
scheduler, which activates every non-faulty robot infinitely often. They have no common notion of time, and their \textsc{Look}, \textsc{Compute}, and
\textsc{Move} phases may take unpredictable amounts of time. Up to $n-2$ robots may crash permanently. A crashed robot stops
executing its \textsc{LCM} cycles but remains visible at its last position.
Since a crash is indistinguishable from an arbitrary delay, the non-faulty
robots cannot rely on the future movement of any particular robot. We characterize the initial configurations from which crash-tolerant optimal
gathering is possible and give a deterministic algorithm that gathers all
non-faulty robots at the same \textsc{Weber Meeting Node} without waiting for
any other robot.

\noindent\textbf{Related work.}
Gathering has been studied in the continuous plane under different assumptions
on visibility, multiplicity detection, orientation, and activation
schedules~\cite{prencipe_impossibility,flocchini_limited_visibility}, and in discrete
environments such as rings, trees, grids, and
hypercubes~\cite{klasing_ring,dangelo_ring,izumi_ring,dangelo_grid_tree,bose_hypercube}. For rings, impossibility and algorithmic
results under different forms of multiplicity detection were given
in~\cite{klasing_ring,dangelo_six,kamei_symmetric}. Gatherable configurations on finite grids and trees were
characterized in~\cite{dangelo_grid_tree}.

Problems with prescribed destinations include landmark covering
~\cite{fujinaga_landmarks}, embedded pattern formation~\cite{cicerone_epf},
and gathering over meeting points~\cite{cicerone_meeting_points}. Bhagat et
al.~\cite{doi:10.3233/FI-222128} studied gathering over meeting nodes on an
infinite grid in the fault-free setting. Fault-tolerant gathering has been
studied in the plane~\cite{agmon_peleg,defago_fault_tolerant,DBLP:conf/icdcs/Bouzid0T13,DBLP:journals/siamcomp/IzumiSKIDWY12,bhagat_one_axis} and on graphs under different faults
~\cite{byzantine_graph,delay_fault,faulty_token,malicious_agent}.
Pelc~\cite{pelc_crash_graph} considered labeled agents with memory and
communication, while Bose et al.~\cite{bose_crash_grid} considered anonymous
robots on finite grids with at most one crash. In contrast, we study crash-tolerant \emph{optimal gathering} over a finite
set of designated meeting nodes on an infinite grid, where the gathering
destination must be a \textsc{Weber Meeting Node}. Our model allows up to
$n-2$ permanent crash faults, which are indistinguishable from arbitrary
activation delays, and therefore does not permit the robots to rely on the
future movement of any particular robot.

\noindent\textbf{Our contributions.}
Our main contributions to crash-tolerant optimal gathering at Weber Meeting Nodes are as follows:
\begin{itemize}
\item We prove that weak multiplicity detection is insufficient for
optimal Weber-meeting-node gathering in the presence of up to $n-2$
crash faults, even under a fully synchronous (\textsc{FSYNC}) scheduler,
where all non-faulty robots are activated simultaneously in each round
and execute their \textsc{Look--Compute--Move} cycles synchronously. Thus, strong multiplicity detection is necessary.
\item We characterize reflectionally and rotationally symmetric
configurations for which optimal gathering becomes impossible when
robots occupying symmetry-breaking positions crash.

\item For every remaining configuration, we give a deterministic algorithm
that selects a unique Weber Meeting Node without relying on the movement
of any distinguished robot.
\item We introduce target-dependent restricted shortest paths and prove
that they preserve the selected region and the target
\textsc{Weber Meeting Node} throughout the execution. This guarantees
that all non-faulty robots gather at the selected target in finite time
despite up to $n-2$ crash faults.
\end{itemize}


\noindent\textbf{Terminology and notation.} We use the infinite-grid model and basic terminology of
\cite{doi:10.1142/S0129054122500174}. Let
$\mathcal G=(\mathcal V,\mathcal E)$ denote the infinite square grid,
$\mathcal M=\{m_1,\ldots,m_s\}\subset\mathcal V$ the finite set of meeting
nodes, and $\mathcal R(t)=\{r_1(t),\ldots,r_n(t)\}$ the multiset of robot
positions at time $t$. The entire configuration is
$\mathcal I(t)=(\mathcal R(t),\mathcal M)$, with initial time $t_0$.
For $u,v\in\mathcal V$, let $d(u,v)$ denote their Manhattan distance, and
let $\lambda_t(v)$ denote the number of robots occupying $v$ at time $t$.
Initially, $\lambda_{t_0}(v)\leq 1$ for every $v\in\mathcal V$. For each $m\in\mathcal M$, define
$c_t(m)=\sum_{v\in\mathcal V}d(v,m)\lambda_t(v)$ and
$\mathcal W(t)=\arg\min_{m\in\mathcal M}c_t(m)$. The nodes in
$\mathcal W(t)$ are called \textsc{Weber Meeting Nodes}. We classify configurations using the reflectional and rotational
symmetries of $\mathcal M$ and $\mathcal I(t)$. Throughout the paper, we
also use traversal strings $\Gamma$, pivot corners, Candidate Weber sets, canonical traversal strings $\alpha$, anchor corners $\mathcal{K}_{a}$, and balanced or unbalanced regions in our target-selection and movement rules. Due to page constraints, the
formal definitions of grid deployment and symmetry, the complete traversal
conventions, an illustrative computation, and the required Weber-node invariance property are provided in the APPENDIX.

\noindent\textbf{Problem Definition.}  Let $f\leq n-2$ denote an upper bound on the number of crash faults that the algorithm can tolerate. For a given execution, let $\mathcal F\subseteq\mathcal R$
denote the set of robots that crash during the execution, where
$|\mathcal F|\leq f$. The set of non-faulty robots is denoted by
$\mathcal R_{nf}=\mathcal R\setminus\mathcal F.$
Consequently, $|\mathcal R_{nf}|=n-|\mathcal F|\geq 2.$
A crashed robot remains permanently stationary at the node where it
crashes. The \emph{crash-tolerant gathering over Weber Meeting Nodes problem}
requires all non-faulty robots to gather at a common
\textsc{Weber Meeting Node} of the initial configuration. More precisely,
an execution starting from $\mathcal I(t_0)$ is successful if there exist
a finite time $T$ and a meeting node, say $m^\dagger\in\mathcal W(t_0)$ such that
$p_r(t)=m^\dagger,
\ \   \text{for every } r\in\mathcal R_{nf} \text{ and every } t\geq T.
$
The crashed robots are not required to reach $m^\dagger$.

\noindent\textbf{Partition of the Initial Configuration.} Let $\mathcal I(t_0)=(\mathcal R(t_0),\mathcal M)$ be an initial configuration, $\mathcal W(t_0)$ the set of \textsc{Weber Meeting Nodes}, and $\mathcal W_c(t_0)$ the set of \textit{Candidate Weber Meeting Nodes}. We partition all initial configurations into the following mutually exclusive and exhaustive classes.

\noindent\textbf{$\mathcal I_1$:}
The configuration admits a unique \textsc{Weber Meeting Node}; that is,
$|\mathcal W(t_0)|=1.$

\noindent\textbf{$\mathcal I_2$:}
The meeting-node configuration $\mathcal M$ is asymmetric and admits
multiple \textsc{Weber Meeting Nodes}; namely,
$|\mathcal W(t_0)|\ge2.$ In this case, there exists a unique \textit{Candidate Weber Meeting Node}, i.e.,$|\mathcal W_c(t_0)|=1.$

\noindent\textbf{$\mathcal I_3$:}
The meeting-node configuration $\mathcal M$ admits a unique line of
symmetry $\mathcal L$ and multiple \textsc{Weber Meeting Nodes}; namely,
$|\mathcal W(t_0)|\geq 2.$ Moreover, the initial robot configuration
$\mathcal R(t_0)$ is asymmetric. The following subclasses are distinguished. \textbf{(a)} \textbf{$\mathcal I_{31}$:} At least one \textsc{Weber Meeting Node} lies on the symmetry line $\mathcal L$; that is, $\mathcal W(t_0)\cap\mathcal L\neq\varnothing.$ This class is further partitioned as follows. \textbf{(a.1)} \textbf{$\mathcal I_{311}$:} If the symmetry line $\mathcal L$ is horizontal or vertical, then exactly one \textit{Candidate Weber Meeting Node} exists, i.e., $|\mathcal W_c(t_0)|=1.$ \textbf{(a.2)} \textbf{$\mathcal I_{312}$:} The symmetry line $\mathcal L$ is diagonal. In this case, either one or two \textit{Candidate Weber Meeting Nodes} exist; that is, $|\mathcal W_c(t_0)|\in\{1,2\}.$ \textbf{(b)} \textbf{$\mathcal I_{32}$:} No \textsc{Weber Meeting Node} lies on the symmetry line $\mathcal L$, and exactly two \textit{Candidate Weber Meeting Nodes} exist; that is, $\mathcal W(t_0)\cap\mathcal L=\varnothing $ \text{and} $|\mathcal W_c(t_0)|=2.$ This class is further partitioned as follows. \textbf{(b.1)} \textbf{$\mathcal I_{321}$:} The symmetry line $\mathcal L$ is horizontal or vertical. \textbf{(b.2)} \textbf{$\mathcal I_{322}$:} The symmetry line $\mathcal L$ is diagonal, and the configuration admits a unique \textit{pivot corner}. \textbf{(b.3)} \textbf{$\mathcal I_{323}$:} The symmetry line $\mathcal L$ is diagonal, and the configuration admits two \textit{pivot corners}.

\noindent\textbf{$\mathcal I_4$:}
The meeting-node configuration $\mathcal M$ admits a nontrivial rotational
symmetry with center $c$ and multiple \textsc{Weber Meeting Nodes}; namely, $|\mathcal W(t_0)|\ge2.$ The initial robot configuration $\mathcal R(t_0)$ is asymmetric. This class is further divided into the following subclasses. \textbf{(a)} \textbf{$\mathcal I_{41}$:} The center of rotation is itself a \textsc{Weber Meeting Node}; that is, $c\in\mathcal W(t_0).$ \textbf{(b)} \textbf{$\mathcal I_{42}$:} The center of rotation is not a \textsc{Weber Meeting Node}, and there exist at least two \textit{Candidate Weber Meeting Nodes}; namely, $c\notin\mathcal W(t_0)$ \text{and} $|\mathcal W_c(t_0)|\ge2.$ This class is further partitioned as follows. \textbf{(b.1)} \textbf{$\mathcal I_{421}$:} At least one Candidate \textsc{Weber Meeting Node} lies in an open quadrant; that is, $\mathcal W_c(t_0)\cap\left(\bigcup_{i=1}^{4}\mathcal Q_i\right)\neq\varnothing.$
This class is further partitioned as follows. \textbf{(b.1.1)} \textbf{$\mathcal I_{4211}$:} There exists a unique quadrant that contains the maximum number of
\textit{Candidate Weber Meeting Node}s; that is, writing $\mu_k=|\mathcal W_c(t_0)\cap\mathcal Q_k|$ and $\mu^{\max}=\max_k\mu_k$, $\big|{k:\mu_k=\mu^{\max}}\big|=1.$ \textbf{(b.1.2)} \textbf{$\mathcal I_{4212}$:}
There exist multiple quadrants that contain the maximum number of
\textit{Candidate Weber Meeting Node}s; that is, $\mathcal S:={k:\mu_k=\mu^{\max}},|\mathcal S|\ge2.$ This class is further partitioned as follows. \textbf{(b.1.2.1)} \textbf{$\mathcal I_{42121}$:} The configuration is \emph{unbalanced}: a unique quadrant in $\mathcal S$ attains $\nu^{\max}_{\mathcal S}$; that is, $\big|{k\in\mathcal S:\nu_k=\nu^{\max}_{\mathcal S}}\big|=1.$ \textbf{(b.1.2.2)} \textbf{$\mathcal I_{42122}$:} The configuration is \emph{balanced}: it is not unbalanced; that is, $\big|{k\in\mathcal S:\nu_k=\nu^{\max}_{\mathcal S}}\big|\ge 2.$ \textbf{(b.2)} \textbf{$\mathcal I_{422}$:} No Candidate \textsc{Weber Meeting Node} lies in any open quadrant; that is, $\mathcal W_c(t_0)\cap\left(\bigcup_{i=1}^{4}\mathcal Q_i\right)=\varnothing.$ This class is further partitioned as follows. \textbf{(b.2.1)} \textbf{$\mathcal I_{4221}$:} The configuration is \emph{unbalanced}: a unique quadrant attains
$\nu^{\max}$; that is, $\big|{k:\nu_k=\nu^{\max}}\big|=1$. \textbf{(b.2.2)} \textbf{$\mathcal I_{4222}$:} The configuration is \emph{balanced}: $\big|{k:\nu_k=\nu^{\max}}\big|\ge 2$.

Let $\mathcal I(t_0)$ denote the set of all initial configurations. Note
that the classes and subclasses defined above form a partition of
$\mathcal I(t_0)$. Since every robot has global visibility of the initial
configuration $\mathcal I(t_0)$, all robots uniquely identify the same class
and subclass during their \textsc{Look} phase. Hence, no ambiguity or
conflict arises in the classification.

\section{Impossibility Results}
\label{sec:impossibility}

The impossibility results are based on a common geometric obstruction:
robots occupying symmetric positions cannot deterministically break the
symmetry when no admissible gathering node is fixed by that
symmetry~\cite{doi:10.1142/S0129054122500174}.

\noindent\textbf{Impossibility inherited from the crash-free setting.}
The partitive-symmetry results of
\cite{doi:10.1142/S0129054122500174} imply that
gathering is impossible from any symmetric initial configuration in which
the fixed part of the symmetry contains no admissible target. In
particular, this includes reflectionally symmetric configurations whose
symmetry line contains no meeting node, and rotationally symmetric
configurations whose center is not a meeting node. For optimal gathering,
the same conclusion holds when the fixed part contains no
\textsc{Weber Meeting Node}. These results remain valid in our model
because every crash-free execution is an admissible execution when up to
$n-2$ robots may crash.


\begin{theorem}[Crash-induced symmetry impossibility]
\label{thm:crash-induced-symmetry}
No deterministic algorithm can guarantee gathering at a meeting node in
any of the following cases: \emph{(i)} $\mathcal I(t_0)$ is reflectionally
symmetric about a line $\mathcal L$,
$\mathcal M\cap\mathcal L=\varnothing$, and all robots initially on
$\mathcal L$ crash before moving; \emph{(ii)} $\mathcal I(t_0)$ is
rotationally symmetric about a grid node $c\notin\mathcal M$, and the
robot initially at $c$ crashes before moving; or \emph{(iii)}
$\mathcal I(t_0)$ is rotationally symmetric about $c\notin\mathcal M$,
and all robots initially on two invariant perpendicular grid lines through
$c$ crash before moving. In each case, at least one complete nontrivial
symmetry orbit outside the crashed set remains non-faulty. Consequently,
gathering at a \textsc{Weber Meeting Node} is also impossible.
\end{theorem}

\begin{theorem}[Impossibility under weak multiplicity detection]
\label{thm:weak-multiplicity-impossible}
No deterministic distributed algorithm equipped only with global weak
multiplicity detection can solve optimal Weber-meeting-node gathering in
the presence of up to $n-2$ crash faults, even under \textsc{FSYNC}
activation and rigid movement.
\end{theorem}

In the \textsc{FSYNC} model, all non-crashed agents are activated
synchronously in each round. Thus, in every round, each non-crashed
agent performs a complete Look--Compute--Move cycle, and all agents
observe the same global configuration before making their moves. The complete proofs of these two theorems are given in the APPENDIX. Therefore, global strong multiplicity detection is necessary, even under \textsc{FSYNC} activation and rigid movement. Let $\mathbb I$ be the set of all initial configurations and let
$\mathbb U\subseteq\mathbb I$ contain the configurations excluded by the
preceding impossibility results. The proposed algorithm is therefore
defined for every initial configuration
$\mathcal I(t_0)\in\mathbb I\setminus\mathbb U$.
\section{Algorithm \textsc{CrashTolerantWeberGathering()}}

\noindent\textbf{High-Level Idea of the Algorithm.}
The algorithm consists of two phases: \emph{Target Selection} and \emph{Crash-Tolerant Directed Gathering}. Upon activation,
a robot observes the current configuration and identifies the applicable
configuration class from the symmetry of $\mathcal M$ and the robot
positions. In the first phase, the robot either selects a common target
\textsc{Weber Meeting Node} $m^\dagger$ directly or first determines a
target region. Direct selection applies when $\mathcal M$ is asymmetric,
when $c\in\mathcal M$ is the unique \textsc{Weber Meeting Node}, or when a
suitable candidate \textsc{Weber Meeting Node} lies on the symmetry line
$\mathcal L$. In the remaining symmetric cases, the prescribed occupancy
and tie-breaking rules determine either a target half-plane
$\mathcal B^\dagger$ or a target quadrant $\mathcal Q^\dagger$. In the
rotational case, the selected region is
$\mathcal Q^\dagger\cup\mathcal L_1^{++}\cup\mathcal L_2^{++}$, where
$\mathcal L_1^{++}$ and $\mathcal L_2^{++}$ bound
$\mathcal Q^\dagger$. A unique node $m^\dagger$ associated with the
selected region is then chosen. In the second phase, every non-faulty robot executes either
\textsc{Direct Gathering} or \textsc{Transfer Plane}. Under
\textsc{Direct Gathering}, the robot invokes \textsc{MoveToDestination()}
to move toward $m^\dagger$. This applies when $m^\dagger$ is selected
directly and also when
$m^\dagger\in\mathcal L_1^{++}\cup\mathcal L_2^{++}$. Under
\textsc{Transfer Plane}, a robot first moves into $\mathcal B^\dagger$ in
the unique-reflection case or into $\mathcal Q^\dagger$ in the rotational
case, and then invokes \textsc{MoveToDestination()} to proceed toward
$m^\dagger$. The overall structure of the algorithm is illustrated in
Figure~\ref{fig:algo-flow}.
\subsection{Phase~I: Target Selection}
In this phase, every robot selects a common target \textsc{Weber Meeting Node} $m^\dagger$, either directly or, when required by symmetry, after first determining a target half-plane or quadrant.

\noindent\textbf{No-region Target Weber Meeting Node Selection:} For the configurations considered in this subsection, a unique target
\textsc{Weber Meeting Node} can be determined directly. Hence, neither a target half-plane nor a target quadrant needs to be selected. The \textit{target Weber meeting node} $m^\dagger$ is determined according to one of the following cases.

     \noindent \textit{Configuration $\mathcal I_1$:} The configuration admits
    a unique \textsc{Weber Meeting Node}, i.e., $|\mathcal W(t_0)|=1$. Every
    robot deterministically identifies this node and selects it as
    $m^\dagger$.

    \begin{figure}[ht!]
    \centering
    \includegraphics[width=0.65\linewidth]{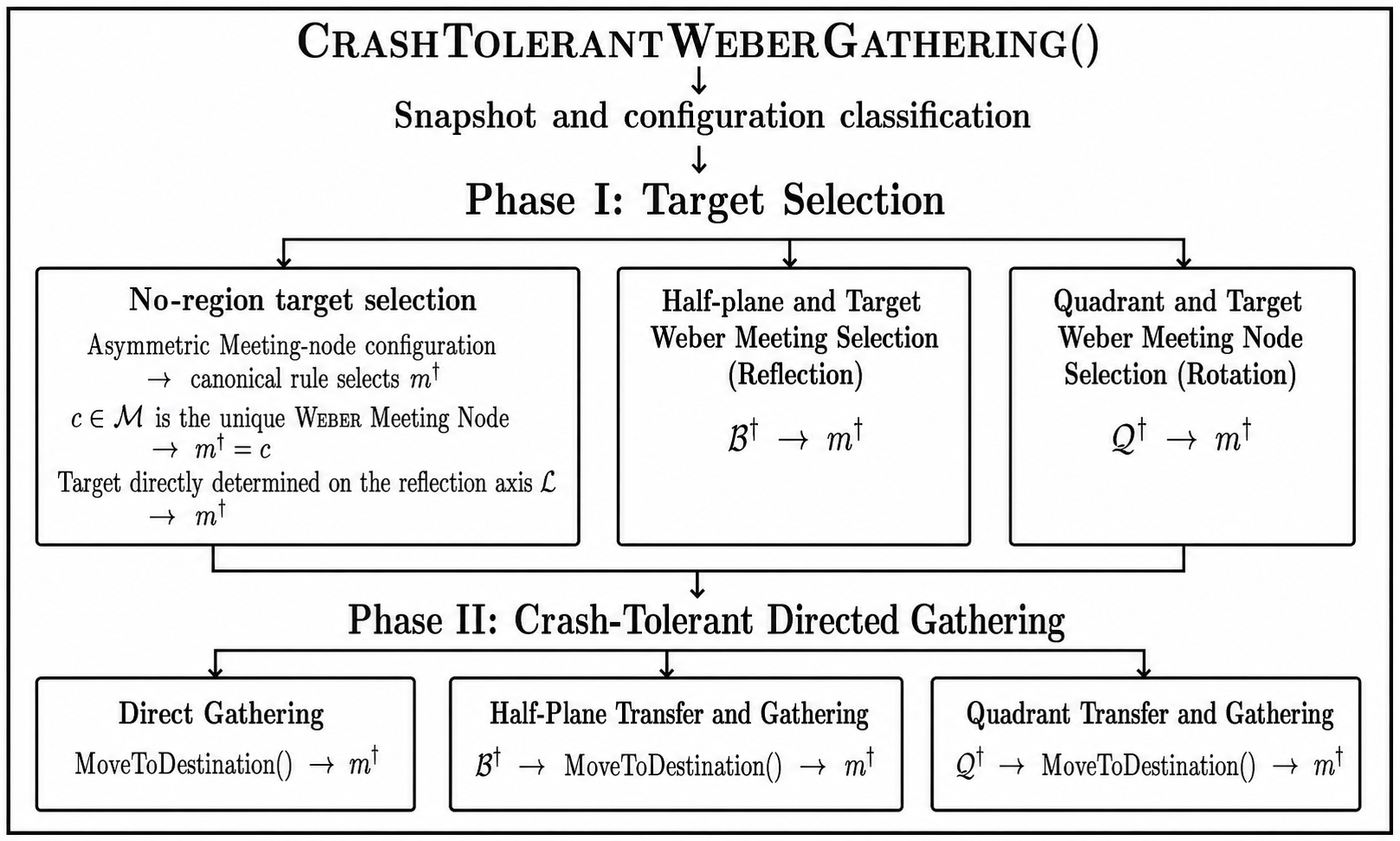}
\caption{Overview of the two phases of the
\textsc{CrashTolerantWeberGathering()} algorithm.}
    \label{fig:algo-flow}
    \vspace*{-0.5cm}
\end{figure}

     \noindent \textit{Configuration $\mathcal I_2$:} The meeting-node
    configuration $\mathcal M$ is asymmetric and admits multiple
    \textsc{Weber Meeting Nodes}, i.e., $|\mathcal W(t_0)|\ge 2$. The unique
    canonical traversal determined by $\mathcal{SER}_{\mathcal M}$ induces
    a canonical ordering of the meeting nodes; the node ranked highest in
    this ordering is uniquely selected as $m^\dagger$.

     \noindent \textit{Configuration $\mathcal I_{311}$:} At least one \textsc{Weber Meeting Node} lies on the symmetry line $\mathcal L$, i.e., $\mathcal W(t_0)\cap\mathcal L\neq\varnothing.$ Moreover, the symmetry line $\mathcal L$ is horizontal or vertical, and there exists exactly one \textit{Candidate Weber Meeting Node}, i.e., $|\mathcal W_c(t_0)|=1.$ This unique candidate node is selected as $m^\dagger$. Since $m^\dagger$ lies on $\mathcal L$, it remains invariant under reflection, and hence every robot independently identifies the same target meeting node.

     \noindent \textit{Configuration $\mathcal I_{312}$:} The symmetry line $\mathcal L$ is diagonal. In this case, either one or two \textit{Candidate Weber Meeting Nodes} exist; that is, $|\mathcal W_c(t_0)|\in\{1,2\}.$ Accordingly, the following two cases arise: If $|\mathcal W_c(t_0)|=1$, then the unique \textit{Candidate Weber Meeting Node} is selected as the target \textsc{Weber Meeting Node} $m^\dagger$. If $|\mathcal W_c(t_0)|=2$, and since the entire configuration $\mathcal I(t_0)$ is asymmetric, then the \textit{Candidate Weber Meeting Node} at the last position of the unique canonical traversal string from the \textit{\textit{anchor corner}} is selected as the target \textsc{Weber Meeting Node} $m^\dagger$. In this configuration, where $|\mathcal W_c(t_0)|=2$, the ambiguity between the two candidates is resolved once at least one of the two or more non-faulty robots moves toward $m^\dagger$. Consequently, the number of candidate \textsc{Weber Meeting Nodes} is reduced to one, i.e., $|\mathcal W_c(t)|=1$.

    \noindent\textit{Configuration $\mathcal I_{41}$:} The meeting-node
    configuration $\mathcal M$ admits a nontrivial rotational symmetry
    with center $c$, and the center of rotation is itself the unique
    \textsc{Weber Meeting Node}, i.e., $c\in\mathcal W(t_0)$. Since $c$ is
    invariant under the rotational symmetry, every robot deterministically
    selects $m^\dagger=c.$

\noindent\textbf{Half-plane and Target Weber Meeting Selection(Reflection):} Suppose that the meeting-node configuration admits a unique line of symmetry $\mathcal L$. The line $\mathcal L$ divides the grid, excluding the nodes on $\mathcal L$, into two open bounded half-planes, denoted by $\mathcal B_1$ and $\mathcal B_2$, which are mapped onto each other under reflection across $\mathcal L$. For configurations in which the \textit{target Weber meeting node} lies outside $\mathcal L$, its selection is performed in two stages. First, a unique target half-plane $\mathcal B^\dagger$ is determined; subsequently, a unique \textit{target Weber meeting node} $m^\dagger$ associated with $\mathcal B^\dagger$ is selected.

\noindent{\textbf{\textit{Phase~I-A (Target Half-Plane Selection):}}} This sub-phase determines the target half-plane
$\mathcal B^\dagger$ whenever the candidate distribution alone does not
directly determine the \textit{target Weber meeting node}.

\noindent\textit{Configuration $\mathcal I_{32}$:}
No \textsc{Weber Meeting Node} lies on the reflection line
$\mathcal L$, and exactly two \textit{Candidate Weber Meeting Nodes}
exist; that is,
$
\mathcal W(t_0)\cap\mathcal L=\varnothing$\ \text{and} $
\left|\mathcal W_c(t_0)\right|=2.
$ The target half-plane is selected according to the following
hierarchical rule.

\noindent\textit{Unequal candidate distribution:}
If $\omega_i>\omega_j$, then $\mathcal B_i$ is the unique half-plane
containing the maximum number of \textit{Candidate Weber Meeting Node}s.
Consequently, $
\mathcal B^\dagger=\mathcal B_i.
$
The other half-plane $\mathcal B_j$ is designated as the transfer half-plane.

\noindent\textit{Equal candidate distribution:}
If $\omega_1=\omega_2$, then, since $\left|\mathcal W_c(t_0)\right|=2$, each half-plane contains exactly one \textit{Candidate Weber Meeting Node}. In this case we have following two cases.

\noindent\textbf{Unbalanced case:}
If $\nu_i>\nu_j$, then the unique half-plane containing the maximum
number of robots is selected:
$
\mathcal B^\dagger=\mathcal B_i.
$
The remaining half-plane $\mathcal B_j$ becomes the transfer
half-plane. This rule applies independently of the orientation of
$\mathcal L$ and the number of pivot corners.

\noindent\textbf{Balanced case:}
If $\nu_1=\nu_2$, then neither the candidate distribution nor the robot occupancy
distinguishes the two half-planes. The tie is therefore resolved
according to the subclass of $\mathcal I_{32}$.

\noindent\textit{Subclass $\mathcal I_{321}$($\mathcal L$ is horizontal or vertical):} Let $\mathcal K_{a}^\dagger$ be the \textit{\textit{anchor corner}} determined by the canonical comparison of
the traversal strings of $\mathcal{SER}_{\mathcal I(t_0)}$. The unique
half-plane not containing $\mathcal K_{a}^\dagger$ is selected as the target half-plane: $\mathcal B^\dagger
=$the unique half-plane not containing $\mathcal K_{a}^\dagger.$

\noindent\textit{Subclass $\mathcal I_{322}$($\mathcal L$ is diagonal):} There is a unique \emph{pivot corner}. The target half-plane is selected by considering the following two
cases.

\noindent\textbf{Case 1}
$(\mathcal{SER}_{\mathcal I(t_0)}
=
\mathcal{SER}_{\mathcal M}$): Let $P$ denote the unique \emph{pivot corner}, and let
$\alpha_{PQ}$ and $\alpha_{PS}$ be the canonical traversal strings
associated with the two traversal directions originating at $P$.
Denote by $\mathcal B_{PQ}$ and $\mathcal B_{PS}$ the corresponding
half-planes. The target half-plane is selected as 
$\mathcal B^\dagger
=
\begin{cases}
\mathcal B_{PS},
&
\text{if }
\alpha_{PS}\succ_{\mathrm{lex}}\alpha_{PQ},
\\[1mm]
\mathcal B_{PQ},
&
\text{if }
\alpha_{PQ}\succ_{\mathrm{lex}}\alpha_{PS}.
\end{cases}$

\noindent\textbf{Case 2}
$(\mathcal{SER}_{\mathcal M}
\subsetneq
\mathcal{SER}_{\mathcal I(t_0)}):$ Let $\mathcal B_1$ and $\mathcal B_2$ denote the two half-planes
induced by $\mathcal L$. For each $i\in\{1,2\}$, define $\eta_i
=
\left|
\left(
\mathcal B_i
\cap
\mathcal{SER}_{\mathcal I(t_0)}
\right)
\cap
\mathbb Z^2
\right|,$ where the grid nodes lying on $\mathcal L$ are excluded from both counts. Thus, $\eta_i$ denotes the number of grid nodes of
$\mathcal{SER}_{\mathcal I(t_0)}$ lying strictly in
$\mathcal B_i$. If $\eta_1\neq\eta_2$, the half-plane containing more grid nodes is
selected as the target half-plane. Hence,
$\mathcal B^\dagger
=
\begin{cases}
\mathcal B_1,
&
\text{if } \eta_1>\eta_2,
\\[1mm]
\mathcal B_2,
&
\text{if } \eta_2>\eta_1.
\end{cases}$

Suppose now that $\eta_1=\eta_2$. Let $\mathcal K_{a}^\dagger$ denote the unique
\emph{anchor corner} of the entire initial configuration
$\mathcal I(t_0)$. Let $\gamma_1$ and $\gamma_2$ denote the two
canonical traversal strings originating at  $\mathcal K_{a}^\dagger$, and let
$\mathcal B_{\gamma_1}$ and $\mathcal B_{\gamma_2}$ denote the
corresponding half-planes. The target half-plane is then selected as
$\mathcal B^\dagger
=
\begin{cases}
\mathcal B_{\gamma_1},
&
\text{if }
\gamma_1\succ_{\mathrm{lex}}\gamma_2,
\\[1mm]
\mathcal B_{\gamma_2},
&
\text{if }
\gamma_2\succ_{\mathrm{lex}}\gamma_1.
\end{cases}$

Thus, when
$\mathcal{SER}_{\mathcal I(t_0)}
=
\mathcal{SER}_{\mathcal M}$,
the target half-plane is determined using the canonical traversal
strings originating at the unique \textit{pivot corner}. When $\mathcal{SER}_{\mathcal M}
\subsetneq
\mathcal{SER}_{\mathcal I(t_0)}$,
the half-plane containing more grid nodes is selected; if the two
grid-node counts are equal, the tie is resolved using the canonical
traversal strings originating at the unique \textit{anchor corner} of
$\mathcal I(t_0)$.

\noindent\textit{Subclass $\mathcal I_{323}$($\mathcal L$ is diagonal):} There are two \emph{pivot corners}. The target half-plane is selected by considering the following two
cases.

\noindent\textbf{Case 1}
$(\mathcal{SER}_{\mathcal I(t_0)}
=
\mathcal{SER}_{\mathcal M}):$ Using the prescribed cyclic labelling $P,Q,R,S$ of the corners of $\mathcal{SER}_{\mathcal I(t_0)}$, let $P$ and $R$ be the two pivot corners. Consider $\alpha_{PS}$ and $\alpha_{RS}$ are the two canonical traversal strings associated with the two corners $P$ and $R$, respectively. Denote by $\mathcal B_{PQ}$ and $\mathcal B_{PS}$ the corresponding half-planes. The target half-plane is selected as
$\mathcal B^\dagger
=
\begin{cases}
\mathcal B_{PS},
&
\text{if }
\alpha_{PS}\succ_{\mathrm{lex}}\alpha_{RS},
\\[1mm]
\mathcal B_{PQ},
&
\text{if }
\alpha_{RS}\succ_{\mathrm{lex}}\alpha_{PS}.
\end{cases}$

\noindent\textbf{Case 2}
$(\mathcal{SER}_{\mathcal M}
\subsetneq
\mathcal{SER}_{\mathcal I(t_0)}):$
If $\eta_1\neq\eta_2$, the half-plane containing more grid nodes is
selected as the target half-plane. Thus,
$\mathcal B^\dagger
=
\begin{cases}
\mathcal B_1,
&
\text{if } \eta_1>\eta_2,
\\[1mm]
\mathcal B_2,
&
\text{if } \eta_2>\eta_1.
\end{cases}$

Suppose that $\eta_1=\eta_2$. Let $\mathcal K_{a}^\dagger$ be the unique \emph{anchor corner} of the entire initial configuration $\mathcal I(t_0)$. For
each half-plane, consider only the canonical traversal string whose
traversal direction lies within that half-plane. Let $\beta_1$ and
$\beta_2$ denote the resulting canonical traversal strings associated
with $\mathcal B_1$ and $\mathcal B_2$, respectively. The target
half-plane is then selected as
$\mathcal B^\dagger
=
\begin{cases}
\mathcal B_1,
&
\text{if }
\beta_1\succ_{\mathrm{lex}}\beta_2,
\\[1mm]
\mathcal B_2,
&
\text{if }
\beta_2\succ_{\mathrm{lex}}\beta_1.
\end{cases}$

Thus, when
$\mathcal{SER}_{\mathcal I(t_0)}
=
\mathcal{SER}_{\mathcal M}$,
the two \textit{pivot-corner} traversal strings lying in the same half-plane are compared to determine the \textit{anchor corner} and, consequently, the target half-plane. When $\mathcal{SER}_{\mathcal M}
\subsetneq
\mathcal{SER}_{\mathcal I(t_0)}$,
the target half-plane is determined by the grid-node counts, and any
remaining tie is resolved using the canonical traversal strings
associated with the unique \textit{anchor corner}. The half-plane different from $\mathcal B^\dagger$ is denoted by
$\mathcal B^{\mathrm{tr}}$.

\noindent{\textbf{\textit{Phase~I-B (Target Weber Meeting Node Selection):}}} Once the target half-plane $\mathcal B^\dagger$ has been determined, define the set of \textit{Candidate Weber Meeting Node}s contained in the selected
half-plane by
$\mathcal W_c^\dagger(t_0)
=
\mathcal W_c(t_0)\cap\mathcal B^\dagger.$
By the target-half-plane selection rule,
$\mathcal W_c^\dagger(t_0)\neq\varnothing$. The \textit{target Weber meeting node}
$m^\dagger$ is selected as follows.

\noindent \textbf{Unique candidate in the target half-plane:} If $\left|\mathcal W_c^\dagger(t_0)\right|=1$, then its unique element is directly selected as the \textit{target Weber meeting node}. Thus, if $\mathcal W_c^\dagger(t_0)=\{m\}$, then $m^\dagger=m$. The selected half-plane contains exactly one \textit{Candidate Weber Meeting Node}.

\noindent \textbf{Multiple candidates in the target half-plane:} Suppose that $\left|\mathcal W_c^\dagger(t_0)\right|>1$. Since, $\left|\mathcal W_c(t_0)\right|=2$, this implies $\left|\mathcal W_c^\dagger(t_0)\right|=2$ under which both
\textit{Candidate Weber Meeting Node}s lie in the selected half-plane
$\mathcal B^\dagger$. Let $\tau^\dagger$ denote the prescribed canonical traversal associated
with $\mathcal B^\dagger$. According to the subclass of
$\mathcal I_{32}$, this traversal is determined as follows: For $\mathcal I_{321}$, $\tau^\dagger$ is the canonical traversal
determined by the \textit{anchor corner} $\mathcal K_{a}^{\dagger}$. For $\mathcal I_{322}$, $\tau^\dagger$ is the traversal from the
unique \textit{pivot corner} $P$ in the direction associated with
$\mathcal B^\dagger$. For $\mathcal I_{323}$, $\tau^\dagger$ is the canonical traversal
associated with the corner pair determining $\mathcal B^\dagger$.
More precisely, it is the traversal represented by $\alpha_{PS}$ when
$\mathcal B^\dagger=\mathcal B_{PS}$, and its reflected counterpart when $\mathcal B^\dagger
=
\mu_{\mathcal L}\bigl(\mathcal B_{PS}\bigr).$

Let $u^\dagger_1,u^\dagger_2,\ldots,u^\dagger_{\ell}$ be the ordered sequence of grid nodes encountered according to $\tau^\dagger$. The robots scan this sequence and select the first encountered node belonging to $\mathcal W_c^\dagger(t_0)$. Formally, define
$j^\dagger
=
\min\left\{
j\in\{1,\ldots,\ell\}
:
u^\dagger_j\in\mathcal W_c^\dagger(t_0)
\right\}.$ The \textit{target Weber meeting node} is then defined as $m^\dagger=u^\dagger_{j^\dagger}.$ \\


\noindent\textbf{Quadrant and Target Weber Meeting Node Selection(Rotation):} When the meeting-node configuration $\mathcal M$ possesses a nontrivial
rotational symmetry about the center $c$ and contains multiple
\textsc{Weber Meeting Nodes} (not lying on $c$), each robot independently identifies the
same target meeting node $m^\dagger$. Whenever a quadrant selection is
required, every robot also consistently chooses the same quadrant
$\mathcal Q^\dagger$ according to the following rules.

\noindent\textbf{\textit{Phase~I-A (Target Quadrant Selection).}}
The target quadrant $\mathcal Q^\dagger$ is selected according to the
applicable configuration class.

\noindent\textbf{Configuration $\mathcal I_{4211}$.}
The unique quadrant containing the maximum number of
\textit{Candidate Weber Meeting Node}s is selected as
$\mathcal Q^\dagger$.

\noindent\textbf{Configurations $\mathcal I_{42121}$ and
$\mathcal I_{4221}$.}
The unique quadrant in $\mathcal S$ attaining
$\nu_{\mathcal S}^{\max}$ is selected as $\mathcal Q^\dagger$.

\noindent\textbf{Configuration $\mathcal I_{42122}$.}
Here, at least two specified quadrants attain
$\nu_{\mathcal S}^{\max}$; hence, the target quadrant is selected
according to the following subclasses.

\noindent\textit{Subclass $\mathcal I_{42122}^{(1)}$.}
The meeting-node configuration admits rotational symmetry of angle
$180^\circ$, and at least one specified quadrant contains both a
\textit{Candidate Weber Meeting Node} and a \textit{pivot corner}.
Among all such quadrants attaining $\nu_{\mathcal S}^{\max}$, the unique
quadrant whose pivot corner is associated with the lexicographically
largest string $\alpha_i$ is selected as $\mathcal Q^\dagger$.

\noindent\textit{Subclass $\mathcal I_{42122}^{(2)}$.}
The meeting-node configuration admits rotational symmetry of angle
$180^\circ$, and none of the specified quadrants containing a
\textit{Candidate Weber Meeting Node} contains a \textit{pivot corner}.
Among all such quadrants attaining $\nu_{\mathcal S}^{\max}$, the unique
quadrant whose \textit{non-pivot corner} is associated with the
lexicographically largest string $\beta_i$ is selected as
$\mathcal Q^\dagger$.

\noindent\textit{Subclass $\mathcal I_{42122}^{(3)}$.}
The meeting-node configuration admits rotational symmetry of angle
$90^\circ$. Among the specified quadrants attaining
$\nu_{\mathcal S}^{\max}$, the unique quadrant whose pivot corner is
associated with the lexicographically largest canonical traversal string
$\alpha_i$ is selected as $\mathcal Q^\dagger$.

\noindent\textbf{Configuration $\mathcal I_{4222}$.}
Here, at least two specified quadrants attain
$\nu_{\mathcal S}^{\max}$; hence, robot occupancy alone cannot determine
$\mathcal Q^\dagger$.

\noindent\textit{Subclass $\mathcal I_{42221}$.}
Among all specified quadrants attaining $\nu_{\mathcal S}^{\max}$, the
unique quadrant whose pivot corner is associated with the
lexicographically smallest canonical traversal string $\alpha_i$ is
selected as $\mathcal Q^\dagger$.

\noindent\textbf{\textit{Phase~I-B (Target Weber Meeting Node Selection):}} In this sub-phase, every robot selects the unique target \textsc{Weber Meeting Node} $m^\dagger$ according to the applicable configuration class.

\noindent\textbf{Configuration $\mathcal I_{41}$:}
The center of rotation is itself a \textsc{Weber Meeting Node}, i.e.,
$c\in\mathcal W(t_0)$. Since $c$ is invariant under the rotational
symmetry, every robot deterministically selects $m^\dagger=c$.

\noindent\textbf{Configurations $\mathcal I_{4211}$,
$\mathcal I_{42121}$, $\mathcal I_{42122}^{(1)}$,
$\mathcal I_{42122}^{(2)}$, and $\mathcal I_{42122}^{(3)}$:}
Let $\mathcal Q^\dagger$ be the unique quadrant selected by the
corresponding quadrant-selection rule, let $p^\dagger$ be the corner of
$\mathcal{SER}_{\mathcal M}$ associated with $\mathcal Q^\dagger$, and
define $\mathcal W^\dagger=\mathcal W_c(t)\cap\mathcal Q^\dagger$. If a unique node of $\mathcal W^\dagger$ minimizes its distance from
$p^\dagger$, it is selected as the target \textsc{Weber Meeting Node};
that is,
$m^\dagger=\underset{m\in\mathcal W^\dagger}{\arg\min}\,
d_M(p^\dagger,m)$. Otherwise, let
$\mathcal W_{\min}^\dagger=
\{m\in\mathcal W^\dagger:
d_M(p^\dagger,m)=
\min_{w\in\mathcal W^\dagger}d_M(p^\dagger,w)\}$
be the set of tied candidates. Among these nodes, the unique node
minimizing its distance from the boundary half-line
$\mathcal L_1^{++}$ is selected; hence,
$m^\dagger=
\underset{m\in\mathcal W_{\min}^\dagger}{\arg\min}\,
d_M(m,\mathcal L_1^{++})$.

\noindent\textbf{Configurations $\mathcal I_{4221}$ and
$\mathcal I_{42221}$:}
Let $\mathcal Q^\dagger$ be the unique quadrant selected by the
corresponding quadrant-selection rule, and let $p^\dagger$ denote its
associated pivot corner. Define
$\overline{\mathcal W}^\dagger=
\mathcal W_c(t)\cap
(\mathcal L_1^{++}\cup\mathcal L_2^{++})$.
If $|\overline{\mathcal W}^\dagger|=1$, its unique element is selected
as $m^\dagger$. If $|\overline{\mathcal W}^\dagger|>1$, the asymmetry of the robot
configuration uniquely determines the direction of the canonical
traversal from $\mathcal K_a^\dagger$. Let
$u_1^\dagger,u_2^\dagger,\ldots,u_\ell^\dagger$ be the ordered sequence
of grid nodes encountered in this traversal, and define
$j^\dagger=\min\{j\in\{1,\ldots,\ell\}:
u_j^\dagger\in\overline{\mathcal W}^\dagger\}$.
The first encountered candidate is then selected; that is,
$m^\dagger=u_{j^\dagger}^\dagger$.

\subsection{Phase~II: Crash-Tolerant Directed Gathering}
Once Phase~I determines the target \textsc{Weber Meeting Node} $m^\dagger$ and, if necessary, a target region, every activated non-faulty robot executes the corresponding gathering rule described below.

\noindent\textbf{Direct Gathering:} This subroutine applies to configurations $\mathcal I_1$, $\mathcal I_2$, $\mathcal I_{31}$, and
$\mathcal I_{41}$. Every activated non-faulty robot not already located
at $m^\dagger$ executes \textsc{MoveToDestination()} and moves one hop
toward $m^\dagger$ along a shortest path. A robot already located at
$m^\dagger$ remains stationary.

\noindent\textbf{Half-Plane Transfer and Gathering(Reflection):} This subroutine is executed for every configuration in
$\mathcal I_{32}$, namely, for the subclasses
$\mathcal I_{321}$, $\mathcal I_{322}$, and $\mathcal I_{323}$.
In each of these cases, Phase~I-A selects a unique target half-plane
$\mathcal B^\dagger$. The other half-plane is denoted by
$\mathcal B^{\mathrm{tr}}$ and is called the \emph{transfer half-plane}.
Phase~I-B subsequently selects
$m^\dagger\in
\mathcal W_c(t_0)\cap\mathcal B^\dagger.$ The movements of robots in the two half-planes occur concurrently and
are determined by their current positions. A robot located in $\mathcal B^\dagger$ executes
\textsc{MoveToDestination()} and performs its next hop toward
$m^\dagger$ along a prescribed shortest path contained in
$\mathcal B^\dagger$. A robot located on the reflection line $\mathcal L$ executes
\textsc{MoveToDestination()} and selects a prescribed restricted
shortest path whose first hop enters $\mathcal B^\dagger$. It then
continues toward $m^\dagger$ along a shortest path contained in
$\mathcal B^\dagger$. 

A robot located in the transfer half-plane
$\mathcal B^{\mathrm{tr}}$ follows a prescribed restricted shortest path
toward $m^\dagger$. It first moves toward $\mathcal L$, subsequently
crosses $\mathcal L$ and enters $\mathcal B^\dagger$, and thereafter
remains in $\mathcal B^\dagger$ while proceeding toward
$m^\dagger$. Thus, its positional transition is
$\mathcal B^{\mathrm{tr}}
\longrightarrow
\mathcal L
\longrightarrow
\mathcal B^\dagger
\longrightarrow
m^\dagger.$ Once such a robot enters $\mathcal B^\dagger$, it follows the
\emph{Direct Gathering} rule at every subsequent activation until it
reaches $m^\dagger$. Every path prescribed by \textsc{MoveToDestination()} is a shortest
path toward the same target $m^\dagger$. Consequently, every completed
hop from $r$ to $r'$ satisfies
$d(r',m^\dagger)=d(r,m^\dagger)-1.$ Meanwhile, robots already belonging to $\mathcal B^\dagger$ continue
moving toward $m^\dagger$ without waiting for robots in
$\mathcal B^{\mathrm{tr}}$ to cross $\mathcal L$. Therefore, gathering
within $\mathcal B^\dagger$ and the transfer of robots from
$\mathcal B^{\mathrm{tr}}$ occur simultaneously rather than as two
separate sequential phases.

\noindent\textbf{Quadrant Transfer and Gathering(Rotation):} 
For rotationally symmetric configurations, every activated robot follows the prescribed movement rule according to whether $m^\dagger$ lies in the interior of $\mathcal Q^\dagger$ or on one of its boundary half-lines, $\mathcal L_1^{++}$ or $\mathcal L_2^{++}$.\\
\noindent{\textit{ \textbf{Case 1( $m^\dagger$ is in the interior of $\mathcal Q^\dagger$):}}} This case applies to configurations in $\mathcal I_{4211},
\mathcal I_{42121},\mathcal I_{42122}^{(1)},\mathcal
I_{42122}^{(2)},\mathcal I_{42122}^{(3)}$ whenever the selected target
$m^\dagger\in\mathcal W^\dagger$ lies strictly inside $\mathcal Q^\dagger$.
Relative to $\mathcal Q^\dagger$, the two quadrants that share a boundary
half-line with it, along $\mathcal L_1$ and $\mathcal L_2$, respectively,
are called the \textit{adjacent} quadrants and are denoted by
$\mathcal Q^{\text{adj}}_1$ and $\mathcal Q^{\text{adj}}_2$. The remaining
quadrant shares no boundary half-line with $\mathcal Q^\dagger$ and
intersects it only at the center of rotation $c$. We refer to this quadrant
as the \textit{opposite} quadrant and denote it by $\mathcal Q^{\text{opp}}$. Three types of movement occur concurrently. (i) a robot already located in $\mathcal Q^\dagger$ executes
\textsc{MoveToDestination()} toward $m^\dagger$ along a shortest path
contained in $\mathcal Q^\dagger$. (ii) a robot located in $\mathcal Q^{\text{adj}}_1$,
$\mathcal Q^{\text{adj}}_2$, or on $\mathcal L_1$ or $\mathcal L_2$,
executes \textsc{MoveToDestination()} along the prescribed restricted
shortest path that crosses directly into, and subsequently remains in,
$\mathcal Q^\dagger$; upon crossing, it continues toward $m^\dagger$ along a
shortest path contained in $\mathcal Q^\dagger$. (iii) a robot located in $\mathcal Q^{\text{opp}}$ shares no boundary
half-line with $\mathcal Q^\dagger$; it therefore executes
\textsc{MoveToDestination()} to move along the prescribed shortest path
toward $c$. Once it reaches $c$, at its next activation it selects the
prescribed restricted shortest path whose first hop enters $\mathcal
Q^\dagger$ across $\mathcal L_1$ or $\mathcal L_2$, and thereafter proceeds
as above. Its local transition is $\mathcal Q^{\text{opp}} \longrightarrow
c \longrightarrow \mathcal L_1\text{ or }\mathcal L_2 \longrightarrow
\mathcal Q^\dagger \longrightarrow m^\dagger$.
\vspace{-0.05cm}

\noindent{\textbf{\textit{\textbf{Case 2 ($m^\dagger$ is on $\mathcal L_1^{++}$ or $\mathcal L_2^\dagger$})}:}}
This case applies to configurations in $\mathcal I_{4221}$ and
$\mathcal I_{42221}$, where $\overline{\mathcal W}^\dagger$ may contain
candidates lying on the boundary half-lines, and whenever the selected
target satisfies $m^\dagger\in\mathcal L_1^{++}\cup\mathcal L_2^{++}$.
Without loss of generality, assume $m^\dagger\in\mathcal L_1^{++}$; the
case $m^\dagger\in\mathcal L_2^{++}$ is symmetric. As in Case~1, let
$\mathcal Q^{\text{adj}}_1,\mathcal Q^{\text{adj}}_2$ denote the two
quadrants adjacent to $\mathcal Q^\dagger$ along $\mathcal L_1$ and
$\mathcal L_2$ respectively, and let $\mathcal Q^{\text{opp}}$ denote the
quadrant diagonally opposite to $\mathcal Q^\dagger$. Four types of movement
occur concurrently. (i) a robot already located in $\mathcal Q^\dagger$ executes
\textsc{MoveToDestination()} toward $m^\dagger$ along a shortest path
contained in $\mathcal Q^\dagger\cup\mathcal L_1^{++}$. (ii) a robot positioned in either $\mathcal Q^{\text{adj}}_1$ or
$\mathcal Q^{\text{adj}}_2$ and sufficiently close to $m^\dagger$ so that
$m^\dagger$ is reachable via a shortest path directly executes
\textsc{MoveToDestination()} toward $m^\dagger$ along such a path. (iii) a robot located in $\mathcal Q^{\text{adj}}_1$ or
$\mathcal Q^{\text{adj}}_2$ that is farthest from $m^\dagger$, for which
no such direct shortest path exists, instead executes
\textsc{MoveToDestination()} to move toward the boundary half-line it
shares with $\mathcal Q^\dagger$ ($\mathcal L_1^{++}$ or $\mathcal
L_2^\dagger$, respectively). Once it reaches that half-line, it continues
toward $m^\dagger$ along a shortest path contained in the half-line (or,
if the half-line is $\mathcal L_2^{++}$, it proceeds via $c$ onto
$\mathcal L_1^{++}$, as in the next case).   (iv) a robot located in $\mathcal Q^{\text{opp}}$ shares no boundary
half-line with $\mathcal Q^\dagger$ or with $\mathcal L_1^{++}$; it
therefore executes \textsc{MoveToDestination()} to move along the
prescribed shortest path toward $c$. Once it reaches $c$, at its next
activation,  it selects the prescribed shortest path along $\mathcal
L_1^\dagger$ toward $m^\dagger$. Its local transition is $\mathcal
Q^{\text{opp}} \longrightarrow c \longrightarrow \mathcal L_1^{++}
\longrightarrow m^\dagger$.



\subsection{Routine \textsc{MoveToDestination()}}
After Phase~I, every robot computes the same target
\textsc{Weber Meeting Node} $m^\dagger$. The routine
\textsc{MoveToDestination()} specifies the path followed by an activated
robot toward $m^\dagger$. Let $r$ be a robot located at the grid node $p_r$.
If $p_r\neq m^\dagger$, then $r$ moves to an adjacent grid node $r'$ satisfying
$d(r',m^\dagger)=d(p_r,m^\dagger)-1$. Thus, every move is one hop along a
Manhattan shortest path toward $m^\dagger$; a robot already located at
$m^\dagger$ remains stationary. The restrictions imposed on the choice of a
shortest path depend on the symmetry of the meeting-node configuration
$\mathcal M$.

\noindent\textbf{Movement when $\mathcal M$ is asymmetric:}
When $\mathcal M$ is asymmetric, Phase~I uniquely determines $m^\dagger$, and
no target half-plane or quadrant needs to be preserved. Therefore, a robot may
follow any Manhattan shortest path from its current position to $m^\dagger$.

\noindent\textbf{Movement when $\mathcal M$ has a line of symmetry:}
Suppose that $\mathcal M$ has a unique line of symmetry $\mathcal L$, which
divides the grid into two open half-planes $\mathcal B_i$ and $\mathcal B_j$.
Phase~I either directly selects a target $m^\dagger$ on $\mathcal L$ or first
selects one half-plane as the target half-plane and subsequently determines
$m^\dagger$ within it. Since multiple shortest paths to the target may exist,
allowing robots to choose any of them may fail to break the tie used to select
the target half-plane. Therefore, the prescribed shortest paths ensure that
every robot progresses toward $m^\dagger$ without entering the non-selected
half-plane, so that the selected target half-plane and $m^\dagger$ remain
unchanged.


\noindent\textit{\textbf{Movement within the target half-plane.}}
Suppose that $r\in\mathcal B_i$. The robot selects a canonical shortest path
to $m_i$ that remains entirely in $\mathcal B_i$; a path intersecting
$\mathcal L$ or entering $\mathcal B_j$ is not used. If $r=m_i$, the robot
does not move.

\noindent\textit{\textbf{Movement toward the target half-plane}.}
Suppose that $r\in\mathcal B_j\cup\mathcal L$. The robot follows a canonical
shortest path to $m_i$ that crosses $\mathcal L$ only once. Before crossing
$\mathcal L$, the path remains in $\mathcal B_j\cup\mathcal L$; afterward,
it remains in $\mathcal B_i$. For a robot in $\mathcal B_j$, the restricted
path is chosen so that
$d(r',m_i)-d(r',m_j)\leq d(r,m_i)-d(r,m_j)$. When the robot is on
$\mathcal L$, its next hop enters $\mathcal B_i$ and satisfies
$d(r',m_i)=d(r,m_i)-1$ and $d(r',m_j)=d(r,m_j)+1$. Hence, every non-faulty
robot starting in $\mathcal B_j$ eventually crosses $\mathcal L$, enters
$\mathcal B_i$, and follows the rule for movement within the target
half-plane.

Let $\mathcal I(t')$ be obtained from $\mathcal I(t)$ after one such hop. Then $c_{t'}(m_i)-c_{t'}(m_j)
=c_t(m_i)-c_t(m_j)
+\bigl(d(r',m_i)-d(r,m_i)\bigr)
-\bigl(d(r',m_j)-d(r,m_j)\bigr)$ and therefore
$c_{t'}(m_i)-c_{t'}(m_j)\leq c_t(m_i)-c_t(m_j)$. For a hop from
$\mathcal L$ into $\mathcal B_i$, the inequality is strict:
$c_{t'}(m_i)-c_{t'}(m_j)=c_t(m_i)-c_t(m_j)-2$. The same hop also increases
the number of robots in the target half-plane. Since every robot crossing
$\mathcal L$ enters $\mathcal B_i$, concurrent movements cannot reverse this
comparison. Figure~\ref{fig:restricted-path} in the APPENDIX illustrates these paths.

\noindent\textbf{Movement when $\mathcal M$ has Rotational Symmetry:}
Suppose that $\mathcal M$ has rotational symmetry about the center $c$.
Phase~I either determines $m^\dagger$ directly or first selects the target
quadrant $\mathcal Q^\dagger$ and then selects $m^\dagger$ within it. In the
latter case, an unrestricted shortest path may pass through a non-selected
quadrant and change the occupancy or candidate comparison used to select
$\mathcal Q^\dagger$ or $m^\dagger$. Therefore, the movement paths are
restricted according to the current position of the robot relative to
$\mathcal Q^\dagger$, its boundary half-lines, and $c$.

\begin{figure}[ht!]
\centering
\begin{subfigure}[t]{0.43\linewidth}
    \centering
    \includegraphics[width=\linewidth]{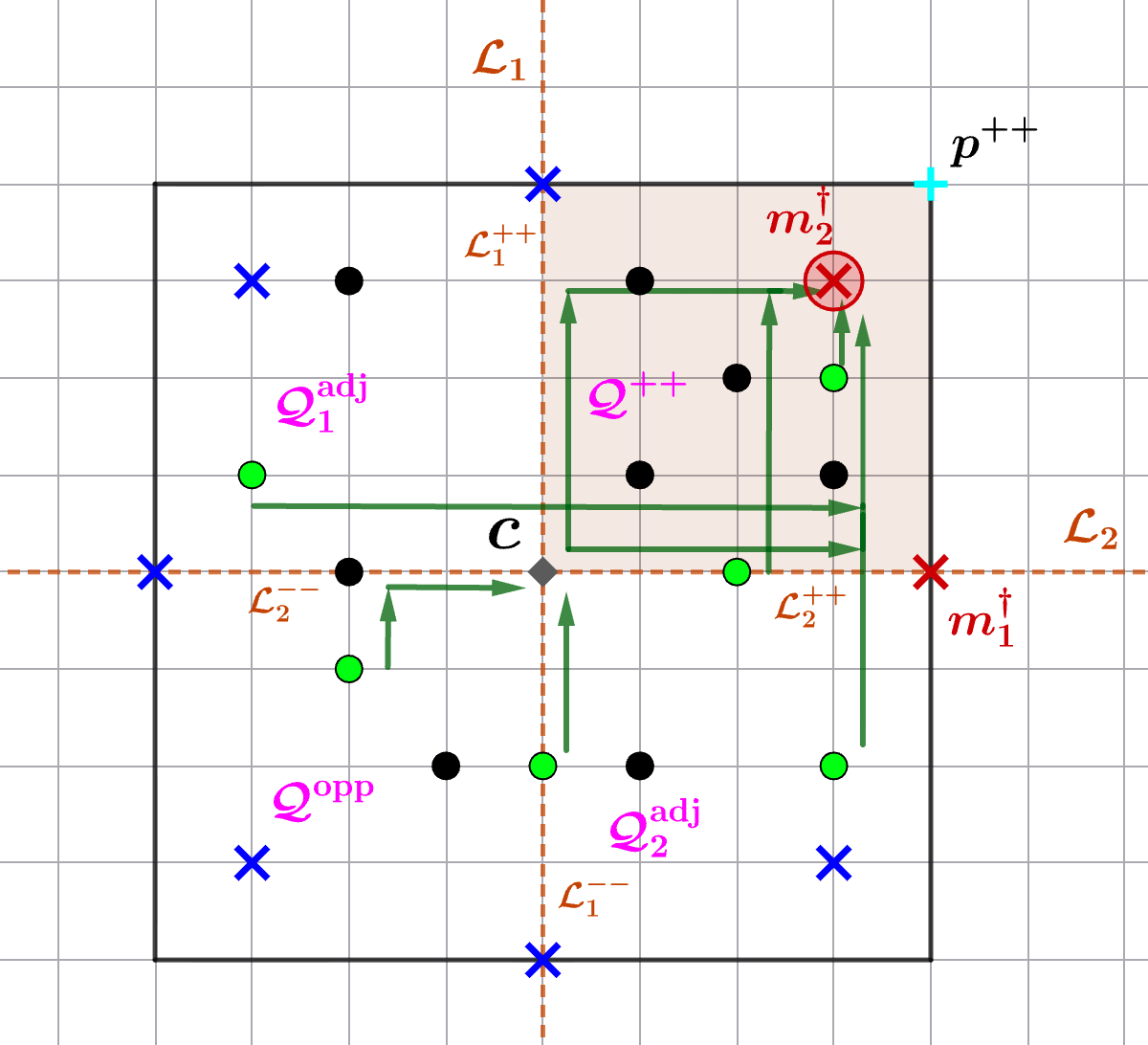}
    \caption{}
\end{subfigure}
\hfill
\begin{subfigure}[t]{0.43\linewidth}
    \centering
    \includegraphics[width=\linewidth]{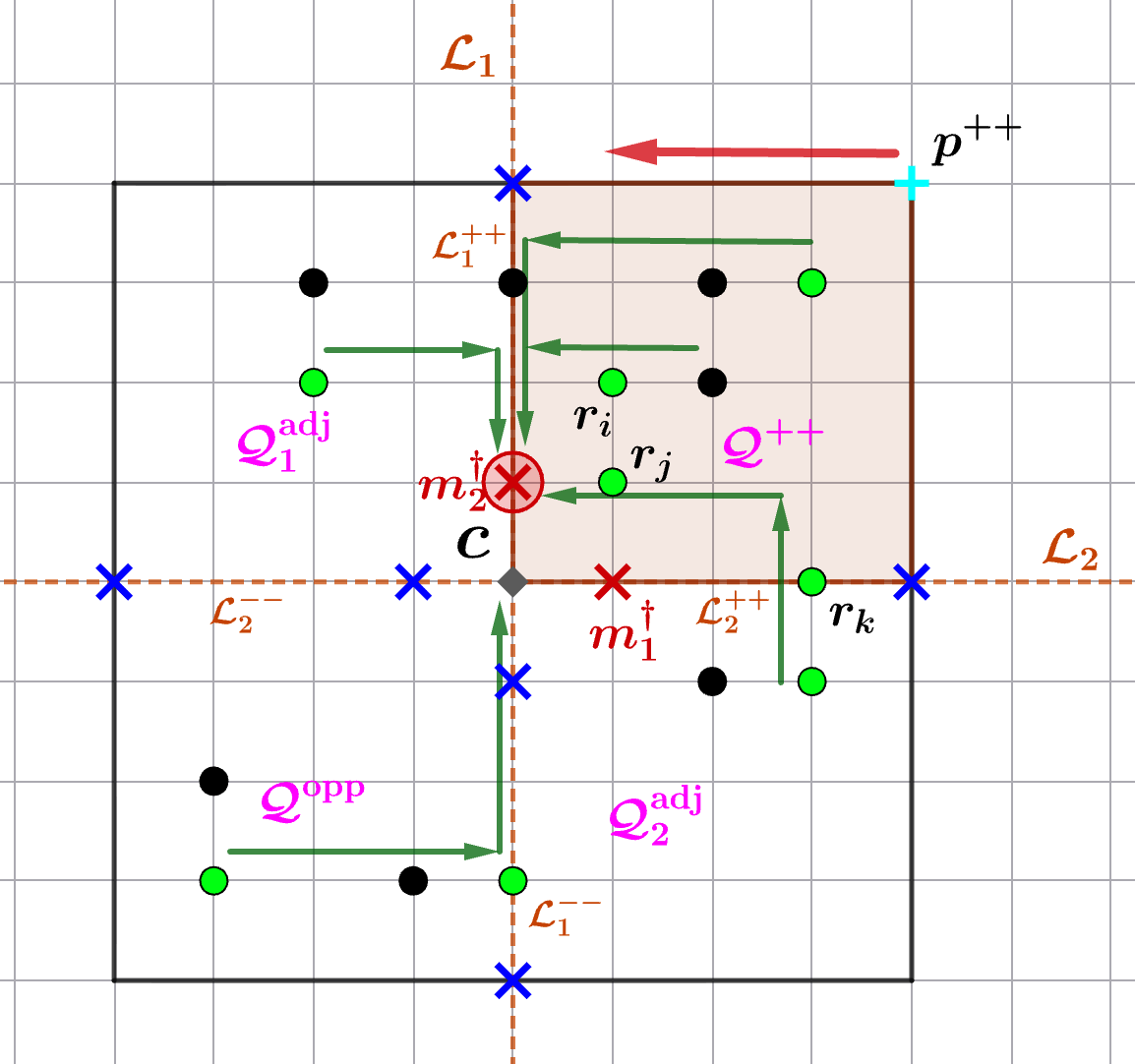}
    \caption{}
\end{subfigure}
\caption{Restricted shortest paths toward $m_2^\dagger$. The green dots
represent the robots activated to move. (a) If
$m_2^\dagger\in\mathcal Q^\dagger$, robots enter $\mathcal Q^\dagger$
directly, except those in $\mathcal Q^{\mathrm{opp}}$, which pass through
$c$. (b) If $m_2^\dagger\in\mathcal L_1^{++}$, robots first reach
$\mathcal L_1^{++}$ through their prescribed regions and then proceed along
it toward $m_2^\dagger$.}
\label{fig:restricted-path_1}
\end{figure}

\noindent\textit{\textbf{Movement toward a \textsc{Weber Meeting Node} in
$\mathcal Q^\dagger$}.}
Let $m_2^\dagger\in\mathcal W^\dagger$ be the target
\textsc{Weber Meeting Node}; see Figure~\ref{fig:restricted-path_1}(a).
Each robot is classified according to its current position relative to
$\mathcal Q^\dagger$.

\noindent\textbf{\textit{Robots already in $\mathcal Q^\dagger$}.}
A robot in $\mathcal Q^\dagger$ executes \textsc{MoveToDestination()} toward
$m_2^\dagger$ along a shortest path contained entirely in
$\mathcal Q^\dagger$.

\noindent\textbf{Robots in the adjacent quadrants
$\mathcal Q_1^{\mathrm{adj}}$ and $\mathcal Q_2^{\mathrm{adj}}$.}
These quadrants share the boundary half-lines $\mathcal L_1^{++}$ and
$\mathcal L_2^{++}$, respectively, with $\mathcal Q^\dagger$. A robot in
either adjacent quadrant follows a restricted shortest path that enters
$\mathcal Q^\dagger$ through the corresponding shared boundary without
passing through $c$, and then follows the rule for robots already in the
target quadrant. Such a robot is never directed through $c$, the other
adjacent quadrant, or $\mathcal Q^{\mathrm{opp}}$.

\noindent\textbf{\textit{Robots in $\mathcal Q^{\mathrm{opp}}$}.}
Since $\mathcal Q^{\mathrm{opp}}$ and $\mathcal Q^\dagger$ meet only at $c$,
a robot in $\mathcal Q^{\mathrm{opp}}$ first follows a shortest path to $c$.
From $c$, it enters $\mathcal Q^\dagger$ through $\mathcal L_1^{++}$ or
$\mathcal L_2^{++}$ and continues toward $m_2^\dagger$; thus,
$\mathcal Q^{\mathrm{opp}}\longrightarrow c\longrightarrow
\mathcal Q^\dagger\longrightarrow m_2^\dagger$. Several robots may meet at
$c$ while following this route. Such a multiplicity causes no ambiguity, as
$c$ lies outside all open quadrants and is equidistant from rotationally
symmetric candidate nodes.

\noindent\textbf{\textit{Robots on $\mathcal L_1^{++}$ or
$\mathcal L_2^{++}$}.}
A robot on either boundary half-line follows a shortest path toward
$m_2^\dagger$. If the path leaves the boundary, it enters
$\mathcal Q^\dagger$ directly; the robot is not routed through $c$.

\noindent\textbf{\textit{Robots on $\mathcal L_1^{--}$ or
$\mathcal L_2^{--}$.}}
A robot on a half-line bounding $\mathcal Q^{\mathrm{opp}}$ first follows a
shortest path to $c$, then enters $\mathcal Q^\dagger$ and proceeds toward
$m_2^\dagger$. For every prescribed hop from $r$ to $r'$,
$d(r',m_2^\dagger)=d(r,m_2^\dagger)-1$. No robot is routed through a
non-target open quadrant other than the quadrant in which it is currently
located. Consequently, the movement does not increase the occupancy of any
non-selected quadrant and does not alter the comparison used to select
$\mathcal Q^\dagger$.

\noindent\textit{\textbf{Movement toward \textit{Candidate Weber Meeting Node}s on
$\mathcal L_1^{++}$ and $\mathcal L_2^{++}$}.}
Without loss of generality, let $m_2^\dagger$ lie on the boundary half-line
$\mathcal L_1^{++}$ of $\mathcal Q^\dagger$; see
Figure~\ref{fig:restricted-path_1}(b). The case
$m_2^\dagger\in\mathcal L_2^{++}$ is symmetric, with the roles of
$\mathcal L_1^{++}$ and $\mathcal L_2^{++}$, and those of
$\mathcal Q_1^{\mathrm{adj}}$ and $\mathcal Q_2^{\mathrm{adj}}$,
interchanged. Here, $\overline{\mathcal W}^\dagger$ contains the selected
candidate $m_2^\dagger$ on $\mathcal L_1^{++}$ and another candidate
$m_1^\dagger$, possibly on $\mathcal L_2^{++}$.

\noindent\textbf{Robots in $\mathcal Q^\dagger$.}
A robot in $\mathcal Q^\dagger$ first follows a shortest path to the nearest
reachable node of $\mathcal L_1^{++}$; except for its endpoint, this path
remains in $\mathcal Q^\dagger$. It then moves along the half-line toward
$m_2^\dagger$.

\noindent\textbf{Robots on $\mathcal L_1^{++}$.}
A robot already on $\mathcal L_1^{++}$ follows the half-line directly toward
$m_2^\dagger$.

\noindent\textbf{Robots in $\mathcal Q_1^{\mathrm{adj}}$.}
Since $\mathcal Q_1^{\mathrm{adj}}$ shares $\mathcal L_1^{++}$ with
$\mathcal Q^\dagger$, a robot in this quadrant follows a shortest grid path
directly to $m_2^\dagger$.

\noindent\textbf{Robots in $\mathcal Q_2^{\mathrm{adj}}$ or on
$\mathcal L_2^{++}$.}
A robot in $\mathcal Q_2^{\mathrm{adj}}$ first follows a restricted shortest
path that enters $\mathcal Q^\dagger$ directly through
$\mathcal L_2^{++}$, moves to $\mathcal L_1^{++}$, and proceeds along it
toward $m_2^\dagger$. A robot on $\mathcal L_2^{++}$ is treated as a robot
in $\mathcal Q_2^{\mathrm{adj}}$.

\noindent\textbf{Robots in $\mathcal Q^{\mathrm{opp}}$.}
A robot in $\mathcal Q^{\mathrm{opp}}$ first follows the prescribed shortest
path to $c$, then moves onto $\mathcal L_1^{++}$ and proceeds along it toward
$m_2^\dagger$; thus,
$\mathcal Q^{\mathrm{opp}}\longrightarrow c\longrightarrow
\mathcal L_1^{++}\longrightarrow m_2^\dagger$. A robot already at $c$
follows the same rule from $c$ onward.

\noindent\textbf{Robots on $\mathcal L_1^{--}$ or
$\mathcal L_2^{--}$.}
A robot on either half-line bounding $\mathcal Q^{\mathrm{opp}}$ first
follows a shortest path to $c$ and then proceeds along
$\mathcal L_1^{++}$ toward $m_2^\dagger$.

For every prescribed hop from $r$ to $r'$,
$d(r',m_2^\dagger)=d(r,m_2^\dagger)-1$, and the restricted paths satisfy
$d(r',m_2^\dagger)-d(r,m_2^\dagger)
\leq d(r',m_1^\dagger)-d(r,m_1^\dagger).$ Let $\mathcal I(t')$ be obtained from $\mathcal I(t)$ after this hop. Then $c_{t'}(m_2^\dagger)-c_{t'}(m_1^\dagger)
=c_t(m_2^\dagger)-c_t(m_1^\dagger)
+\bigl(d(r',m_2^\dagger)-d(r,m_2^\dagger)\bigr)
-\bigl(d(r',m_1^\dagger)-d(r,m_1^\dagger)\bigr)$ and therefore
$c_{t'}(m_2^\dagger)-c_{t'}(m_1^\dagger)
\leq c_t(m_2^\dagger)-c_t(m_1^\dagger)$. The inequality becomes strict
when a robot makes a committed move toward $m_2^\dagger$ along
$\mathcal L_1^{++}$, or enters the side associated with
$\mathcal L_1^{++}$ from $c$ or $\mathcal Q^\dagger$. Consequently, if
$c_t(m_2^\dagger)=c_t(m_1^\dagger)$, the first such successful hop gives
$c_{t'}(m_2^\dagger)<c_{t'}(m_1^\dagger)$. The candidate $m_1^\dagger$ is
then no longer a \textsc{Weber Meeting Node}, whereas $m_2^\dagger$ remains
one; all subsequent movements preserve this strict inequality.

\section{Correctness of the Algorithm}

For a snapshot $\mathcal I(t)=(\mathcal R(t),\mathcal M)$, let
$m^\dagger(t)$ denote the target \textsc{Weber Meeting Node} computed by an
activated robot and, when applicable, let $\mathcal B^\dagger(t)$ or
$\mathcal Q^\dagger(t)$ denote the computed target region. Let $\mathcal F$
be the unknown, possibly empty, set of crash-faulty robots. Due to page constraints, the proofs of the following results are provided in the APPENDIX.

\begin{lemma}[Uniqueness of the target region]
\label{lem:region-unique}
For every configuration in $\mathcal I_{321}$, $\mathcal I_{322}$, or
$\mathcal I_{323}$, Phase~I-A uniquely selects one of $\mathcal B_1$ and
$\mathcal B_2$. For every configuration in
$\mathcal I_{4211}$, $\mathcal I_{42121}$,
$\mathcal I_{42122}^{(1)}$, $\mathcal I_{42122}^{(2)}$,
$\mathcal I_{42122}^{(3)}$, $\mathcal I_{4221}$, or
$\mathcal I_{42221}$, it uniquely selects one of
$\mathcal Q_1,\ldots,\mathcal Q_4$. 
\end{lemma}


\begin{lemma}[Uniqueness of the target node]
\label{lem:node-unique}
For every solvable configuration, Phase~I determines exactly one target
\textsc{Weber Meeting Node} $m^\dagger\in\mathcal W(t)$ from the current
snapshot.
\end{lemma}



\begin{lemma}[Invariance of the selected half-plane]
\label{lem:halfplane-preservation}
Let $\mathcal B^\dagger$ be the target half-plane selected from a
configuration in $\mathcal I_{321}$, $\mathcal I_{322}$, or
$\mathcal I_{323}$. If every subsequent move follows
\textsc{MoveToDestination()}, every configuration reached during Phase~II
selects the same half-plane $\mathcal B^\dagger$.
\end{lemma}


\begin{lemma}[Invariance of the selected quadrant]
\label{lem:quadrant-preservation}
Let $\mathcal Q^\dagger$ be the target quadrant selected at the beginning of
Phase~II. Every subsequent configuration reached during Phase~II selects the
same quadrant $\mathcal Q^\dagger$.
\end{lemma}


\begin{lemma}[Invariance of the selected target]
\label{lem:target-invariant}
Let $m^\dagger$ be selected from $\mathcal I(t)$. If every move completed
after time $t$ follows \textsc{MoveToDestination()} toward $m^\dagger$, then
every subsequent snapshot selects the same target $m^\dagger$.
\end{lemma}


\begin{lemma}[Independent progress under crash faults]
\label{lem:crash-noninterference}
Let $r\notin\mathcal F$ be a non-faulty robot activated at time $t$, with
$r(t)\neq m^\dagger$, and let $t'>t$ be the time at which it completes
the move prescribed by \textsc{MoveToDestination()}. Then
$d(r(t'),m^\dagger)\leq d(r(t),m^\dagger)-1$. The completion of this move
does not require the activation or movement of any other robot.
\end{lemma}


\begin{lemma}[Finite-time gathering]
\label{lem:finite-time-gathering}
Under every fair asynchronous schedule, all non-faulty robots reach
$m^\dagger$ after finitely many completed moves and remain there.
\end{lemma}


\begin{theorem}[Correctness]
\label{thm:correctness}
Let $\mathcal I(t_0)\notin\mathbb U$, and suppose that at most $n-2$
robots crash. Under every fair asynchronous schedule,
\textsc{CrashTolerantWeberGathering()} gathers all non-faulty robots in
finite time at a common node $m^\dagger\in\mathcal W(t_0)$. Moreover, no
non-faulty robot must wait for another robot to move or become activated.
\end{theorem}

\begin{proof}
Because $\mathcal I(t_0)\notin\mathbb U$, the applicable rules of Phase~I
are defined. Lemmas~\ref{lem:region-unique} and \ref{lem:node-unique}
provide a unique $m^\dagger\in\mathcal W(t_0)$ and, when necessary, a unique
target region. Lemma~\ref{lem:target-invariant} shows that every subsequent
snapshot yields the same choices. Lemma~\ref{wbn} (see APPENDIX) additionally ensures that $m^\dagger$ remains a \textsc{Weber Meeting Node} throughout Phase~II.

A crash only prevents the faulty robot from moving further. By
Lemma~\ref{lem:crash-noninterference}, it neither blocks another robot nor
forces a non-faulty robot to wait. Finally,
Lemma~\ref{lem:finite-time-gathering} implies that fairness brings every
non-faulty robot to $m^\dagger$ in finite time and that all such robots remain
there.
\end{proof}

\section{Conclusion and Future Work}

We studied crash-tolerant optimal gathering at \textsc{Weber Meeting Nodes} for asynchronous, anonymous, oblivious, and disoriented robots on an infinite grid, assuming global visibility, global strong multiplicity detection, and up to $n-2$ crashes. We established both inherited and crash-specific impossibility results and showed that global weak multiplicity detection cannot preserve min-sum optimality, even under \textsc{FSYNC} activation and rigid movements. For every remaining configuration, \textsc{CrashTolerantWeberGathering()} selects a unique target $m^\dagger\in\mathcal W(t_0)$ and, when required, a unique target region.Restricted shortest paths preserve these choices and ensure that all non-faulty robots gather at the common target in finite time under any fair \textsc{ASYNC} scheduler. 

Future work includes characterizing solvable configurations under weaker multiplicity detection, possibly using persistent memory or visible lights; removing the grid-node restriction on symmetry lines and rotation centers; and considering limited or obstructed visibility, Byzantine faults, and alternative optimization criteria such as maximum travel distance or completion time.

\bibliographystyle{splncs04}
\bibliography{Biblo}

\newpage

\section*{APPENDIX}

\section*{Detailed Definitions and Notation}
\label{app:definitions}

\noindent\textbf{Grid Deployment and Symmetry.}
Let $\mathcal P=(\mathbb Z,\mathcal E')$ be the infinite path, where
$\mathcal E'=\{\{i,i+1\}:i\in\mathbb Z\}$, and let
$\mathcal G=\mathcal P\times\mathcal P=(\mathcal V,\mathcal E)$ be the
infinite grid. The finite set of meeting nodes is
$\mathcal M=\{m_1,\ldots,m_s\}\subset\mathcal V$, and
$\mathcal R(t)=\{r_1(t),\ldots,r_n(t)\}$ is the multiset of robot positions
at time $t$. The entire configuration is
$\mathcal I(t)=(\mathcal R(t),\mathcal M)$, with initial time $t_0$.
For $u,v\in\mathcal V$, $d(u,v)$ denotes their Manhattan distance. Let $\lambda_t(v)$ be the number of robots at $v$ and let
$f_t(v)=1$ if $v\in\mathcal M$, and $f_t(v)=0$ otherwise.
Although the set of meeting nodes is fixed and hence $f_t$ is independent
of time, we retain the subscript $t$ for notational consistency with the
time-dependent configuration $\mathcal I(t)$. Note that initially,
$\lambda_{t_0}(v)\leq 1$ for every $v\in\mathcal V$.

A grid automorphism $\Phi$ is a configuration automorphism at time $t$ if
$f_t(v)=f_t(\Phi(v))$ and
$\lambda_t(v)=\lambda_t(\Phi(v))$ for every $v\in\mathcal V$. Since the
configuration is finite, only reflectional and rotational symmetries are
relevant. Reflection lines may be horizontal, vertical, or diagonal, and
rotations may be by $90^\circ$ or $180^\circ$ whenever they preserve the
grid. We use partitive automorphisms and their orbits in the sense of
\cite{doi:10.3233/FI-222128}.

\noindent\textbf{Weber Meeting Nodes.}
For each $m\in\mathcal M$, define
$c_t(m)=\sum_{v\in\mathcal V} d(v,m)\lambda_t(v)$.
A meeting node minimizing $c_t$ is called a \textsc{Weber Meeting Node}, and
$\mathcal W(t)=\arg\min_{m\in\mathcal M} c_t(m)$
denotes the set of all such nodes. Thus, any node in $\mathcal W(t_0)$
minimizes the sum of the initial Manhattan distances from all robots to a
designated meeting node. Under global strong multiplicity detection, every
robot can determine $\lambda_t(v)$ exactly at every grid node and hence
compute $\mathcal W(t)$ from its snapshot.

\noindent\textbf{Traversal strings and pivot corners.}
Let $\mathcal{SER}_{\mathcal M}=PQRS$ be the smallest axis-aligned rectangle
containing $\mathcal M$, and let $\mathcal{SER}_{\mathcal I(t)}$ be the
smallest axis-aligned rectangle containing $\mathcal R(t)\cup\mathcal M$. Following the traversal construction of
\cite{doi:10.1142/S0129054122500174}, a corner together with one of its
incident directions determines an ordered scan of the rectangle. The hop-distances
of the meeting nodes encountered during this scan form a traversal string,
denoted by $\Gamma_{PQ}$ when the scan starts at corner $P$ in the
direction of side $PQ$. For a non-square rectangle, only the traversal
parallel to the shorter side is considered, whereas for a square, both
directions are considered. In the collinear case, traversals from the two
opposite ends are used.

All relevant traversal strings of $\mathcal{SER}_{\mathcal M}$ are compared
lexicographically. A corner producing a lexicographically smallest string
is called a \emph{pivot corner}, and its associated scan is its
\emph{traversal direction}. Thus, our pivot corner is the leading corner
of \cite{doi:10.1142/S0129054122500174}. Let
$\mathcal K_{\mathrm{piv}}$ be the set of pivot corners.  For each
$p\in\mathcal K_{\mathrm{piv}}$, let $w_p(t)$ be the last node of
$\mathcal W(t)$ encountered in the traversal associated with $p$. The
\emph{candidate Weber set}, corresponding to the Potential Weber meeting
nodes of \cite{doi:10.1142/S0129054122500174}, is defined as
$\mathcal W_c(t)=\{w_p(t):p\in\mathcal K_{\mathrm{piv}}\}$.

The following invariance property, proved for Weber points in
\cite{inbook} and adapted to meeting nodes in
\cite{doi:10.1142/S0129054122500174}, will be used throughout.

\begin{lemma}
\label{wbn}
Let $m^\dagger\in\mathcal W(t)$, and let $\mathcal I(t')$ be obtained from
$\mathcal I(t)$ after a robot, or a group of co-located robots, moves toward
$m^\dagger$ along a shortest path. Then
$m^\dagger\in\mathcal W(t')$ and
$\mathcal W(t')\subseteq\mathcal W(t)$.
\end{lemma}
Thus, moving toward $m^\dagger$ preserves it as a
\textsc{Weber Meeting Node} and creates no new
\textsc{Weber Meeting Node}.

\begin{observation}
\label{m3}
If the rotation center $c$ of $\mathcal I(t_0)$ is a meeting node, then
$c\in\mathcal W(t_0)$.
\end{observation}

\noindent\textbf{Canonical traversal strings and anchor corner.}
To distinguish configurations having identical meeting-node traversal
strings, each node $u$ encountered by a prescribed traversal $\tau$ is
encoded by $\langle f_t(u),\lambda_t(u)\rangle$. These pairs are ordered by
first comparing $\lambda_t$ and then $f_t$, and the corresponding canonical
traversal string is
$\bigl(\langle f_t(u_1^\tau),\lambda_t(u_1^\tau)\rangle,\ldots,
\langle f_t(u_\ell^\tau),\lambda_t(u_\ell^\tau)\rangle\bigr)$.
A canonical string originating from a pivot corner is called an
$\boldsymbol{\alpha}$-string; one originating from a non-pivot corner is
called a $\boldsymbol{\beta}$-string. The first and second letters in a
subscript specify the starting corner and initial traversal direction,
respectively, as in $\boldsymbol{\alpha}_{PS}(t)$.

The admissible canonical strings of
$\mathcal{SER}_{\mathcal I(t_0)}$ are compared lexicographically. If
$\mathcal I(t_0)$ is asymmetric, their unique minimum originates from the
\emph{anchor corner}; its string is denoted by
$\boldsymbol{\alpha}_{\mathrm{anchor}}(t_0)$. This notion plays the role of
the key corner of \cite{doi:10.1142/S0129054122500174}, while explicitly
using the smallest enclosing rectangle of the entire configuration. 

\noindent\textbf{Regions and balance.}
If $\mathcal M$ has a unique reflection line $\mathcal L$, let
$\mathcal H_1,\mathcal H_2$ be the two open half-planes determined by
$\mathcal L$, and let
$\mathcal B_i=\mathcal H_i\cap\mathcal{SER}_{\mathcal I(t_0)}$.
Define $\omega_i=|\mathcal W_c(t_0)\cap\mathcal B_i|$ and
$\nu_i=\sum_{u\in\mathcal B_i}\lambda_{t_0}(u)$. Thus, $\omega_i$ counts
the candidate Weber meeting nodes and $\nu_i$ counts the robots in
$\mathcal B_i$; objects on $\mathcal L$ are excluded. The configuration is
\emph{balanced} if $\nu_1=\nu_2$, and \emph{unbalanced} otherwise.

If $\mathcal M$ has rotational symmetry about $c$, let
$\mathcal L_1,\mathcal L_2$ be the perpendicular grid lines through $c$,
parallel to the sides of $\mathcal{SER}_{\mathcal M}=PQRS$. These lines
determine the open quadrants $\mathcal Q_1,\ldots,\mathcal Q_4$ containing,
respectively, $p_1=S,p_2=P,p_3=Q,p_4=R$. Let
$\mathcal B_k=\mathcal Q_k\cap\mathcal{SER}_{\mathcal I(t_0)}$,
$\omega_k=|\mathcal W_c(t_0)\cap\mathcal B_k|$, and
$\nu_k=\sum_{u\in\mathcal B_k}\lambda_{t_0}(u)$. Nodes on
$\mathcal L_1\cup\mathcal L_2$ are excluded from these counts. Set
$\omega^{\max}=\max_k\omega_k$,
$\mathcal S=\{k:\omega_k=\omega^{\max}\}$,
$\nu_{\mathcal S}^{\max}=\max_{k\in\mathcal S}\nu_k$, and
$\mathcal S^{\max}=\{k\in\mathcal S:\nu_k=\nu_{\mathcal S}^{\max}\}$.
The quadrants indexed by $\mathcal S$ are called the \emph{specified
quadrants}. The configuration is \emph{balanced} if
$|\mathcal S^{\max}|\geq2$, and \emph{unbalanced} otherwise.

\subsection*{Illustrative Computation}
\label{app:pivot-anchor-computation}

This APPENDIX gives the detailed computation for the configuration in
Figure~\ref{fig:pivot-anchor-example}. Taking the marked grid node as
$(0,0)$ and one grid edge as one hop, the robot positions are
\[
\begin{aligned}
r_1&=(-4,-4),&r_2&=(-2,-4),&r_3&=(2,-4),\\
r_4&=(4,-2),&r_5&=(-2,0),&r_6&=(2,4),
\end{aligned}
\]
and the meeting nodes are
\[
\begin{aligned}
m_1&=(-4,0),&m_2&=(-2,-2),&m_3&=(2,-2),\\
m_4&=(4,0),&m_5&=(-2,2),&m_6&=(2,2),&m_7&=(4,2).
\end{aligned}
\]

\begin{figure}[ht]
\centering
\includegraphics[width=0.5\columnwidth]{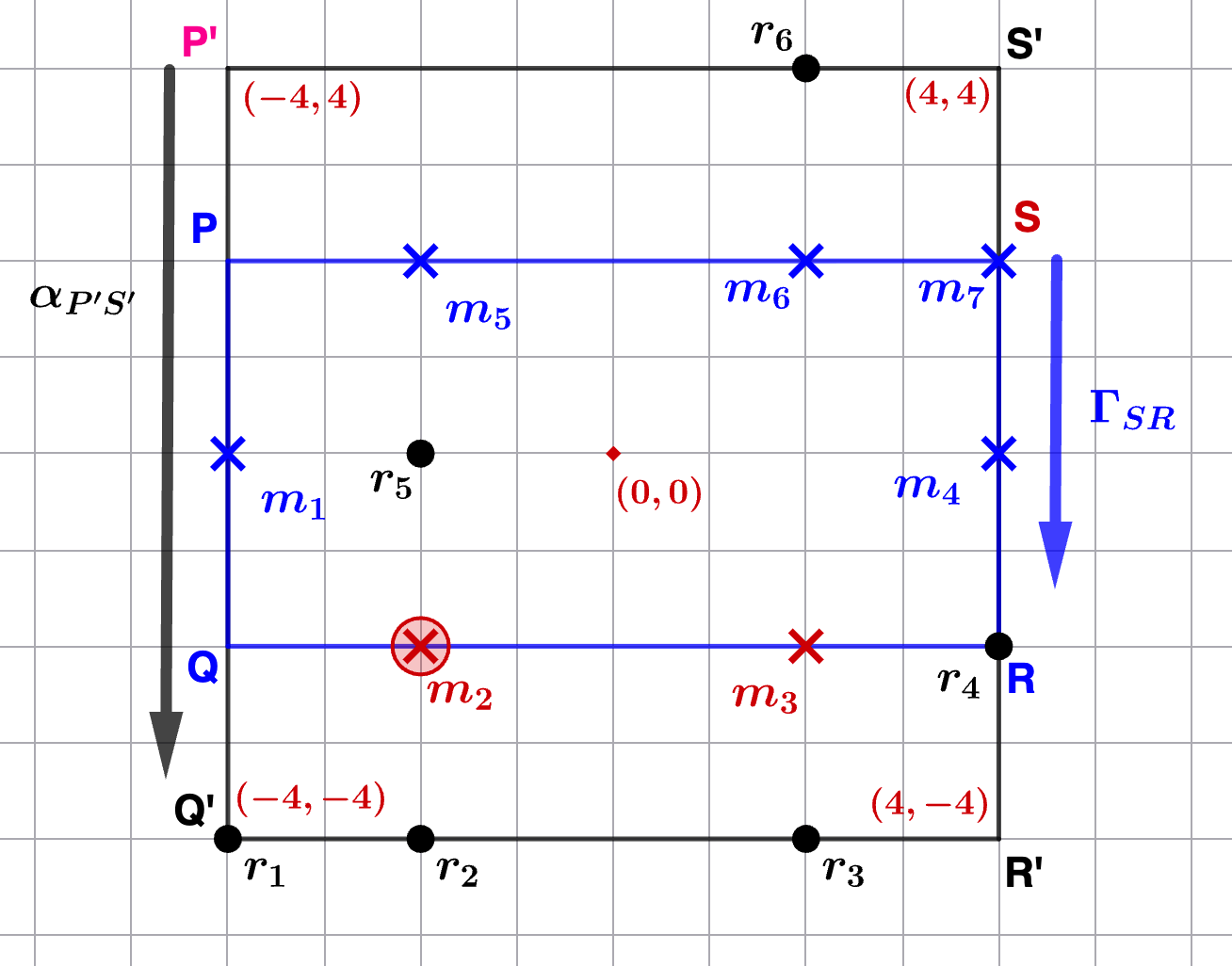}
\caption{An asymmetric configuration with
$\mathcal{SER}_{\mathcal M}=PQRS$ and
$\mathcal{SER}_{\mathcal I}=P'Q'R'S'$.
Here, $\mathcal W(t_0)=\{m_2,m_3\}$. Under the shorter-side traversal
rule for the non-square rectangle $\mathcal{SER}_{\mathcal M}$,
$\Gamma_{SR}$ is the unique lexicographically smallest relevant
$\Gamma$-string; hence, $S$ is the pivot corner and
$\mathcal W_c(t_0)=\{m_2\}$. Moreover, $\alpha_{P'S'}$ is the unique
lexicographically smallest canonical traversal string of the entire
configuration, so $P'$ is the anchor corner.}
\label{fig:pivot-anchor-example}
\end{figure}

The two enclosing rectangles are
\[
\begin{aligned}
\mathcal{SER}_{\mathcal M}&=PQRS=[-4,4]\times[-2,2],\\
\mathcal{SER}_{\mathcal I}&=P'Q'R'S'=[-4,4]\times[-4,4].
\end{aligned}
\]
Consequently,
\[
\begin{aligned}
P&=(-4,2),&Q&=(-4,-2),&R&=(4,-2),&S&=(4,2),\\
P'&=(-4,4),&Q'&=(-4,-4),&R'&=(4,-4),&S'&=(4,4).
\end{aligned}
\]

\paragraph{Weber costs.}
The aggregate costs of the seven meeting nodes are
\[
\begin{aligned}
c_{t_0}(m_1)&=4+6+10+10+2+10=42,\\
c_{t_0}(m_2)&=4+2+6+6+2+10=30,\\
c_{t_0}(m_3)&=8+6+2+2+6+6=30,\\
c_{t_0}(m_4)&=12+10+6+2+6+6=42,\\
c_{t_0}(m_5)&=8+6+10+10+2+6=42,\\
c_{t_0}(m_6)&=12+10+6+6+6+2=42,\\
c_{t_0}(m_7)&=14+12+8+4+8+4=50.
\end{aligned}
\]
Therefore,
\[
\boxed{\mathcal W(t_0)=\{m_2,m_3\}}.
\]

\paragraph{The eight $\Gamma$-strings of $\mathcal{SER}_{\mathcal M}$.}
The rectangle contains $9\times5=45$ grid nodes. Assigning scan index
$1$ to the starting corner gives
\begin{center}
\resizebox{\linewidth}{!}{%
\begin{tabular}{c|l|l}
$\tau$ & \multicolumn{1}{c|}{Order of the meeting nodes}
       & \multicolumn{1}{c}{$\Gamma_\tau$}\\ \hline
PQ&$(m_1,m_5,m_2,m_6,m_3,m_7,m_4)$&$(3,11,15,31,35,41,43)$\\
PS&$(m_5,m_6,m_7,m_1,m_4,m_2,m_3)$&$(3,7,9,19,27,39,43)$\\
QP&$(m_1,m_2,m_5,m_3,m_6,m_4,m_7)$&$(3,11,15,31,35,43,45)$\\
QR&$(m_2,m_3,m_1,m_4,m_5,m_6,m_7)$&$(3,7,19,27,39,43,45)$\\
RQ&$(m_3,m_2,m_4,m_1,m_7,m_6,m_5)$&$(3,7,19,27,37,39,43)$\\
RS&$(m_4,m_7,m_3,m_6,m_2,m_5,m_1)$&$(3,5,11,15,31,35,43)$\\
SP&$(m_7,m_6,m_5,m_4,m_1,m_3,m_2)$&$(1,3,7,19,27,39,43)$\\
SR&$(m_7,m_4,m_6,m_3,m_5,m_2,m_1)$&$(1,3,11,15,31,35,43)$\\
\end{tabular}}
\end{center}
Since $\mathcal{SER}_{\mathcal M}$ has dimensions $8\times4$, the
non-square-rectangle rule retains only the traversals parallel to its
shorter side, namely $PQ,QP,RS$, and $SR$. Among these four relevant
strings,
\[
\Gamma_{SR}<\Gamma_{PQ},\Gamma_{QP},\Gamma_{RS}.
\]
Hence, $S$ is the unique pivot corner and $SR$ is its traversal
direction. Along $SR$, the two \textsc{Weber Meeting Nodes} occur in the
order $m_3\prec m_2$. Since $m_2$ is the last one encountered,
\[
\boxed{\mathcal W_c(t_0)=\{m_2\}}.
\]
Although $\Gamma_{SP}<\Gamma_{SR}$ when all eight strings are compared,
$SP$ is not relevant to the pivot comparison because it is parallel to
the longer side of the non-square rectangle.

\paragraph{The eight $\alpha$-strings of $\mathcal{SER}_{\mathcal I}$.}
For each traversal $\tau$, let
\[
\alpha_\tau=
\Bigl(
\langle f_{t_0}(u_1^\tau),\lambda_{t_0}(u_1^\tau)\rangle,
\ldots,
\langle f_{t_0}(u_{81}^\tau),\lambda_{t_0}(u_{81}^\tau)\rangle
\Bigr),
\]
where
\[
f_{t_0}(u)=
\begin{cases}
1,&u\in\mathcal M,\\
0,&u\notin\mathcal M,
\end{cases}
\]
and $\lambda_{t_0}(u)$ is the number of robots occupying $u$. Although
each entry is written as $\langle f_{t_0}(u),\lambda_{t_0}(u)\rangle$,
the comparison first considers $\lambda_{t_0}(u)$ and then
$f_{t_0}(u)$. The symbol $\mid$ below only separates consecutive
parallel scan lines.

\begingroup
\scriptsize
\[
\begin{aligned}
\alpha_{P'Q'}={}&(\langle0,0\rangle,\langle0,0\rangle,\langle0,0\rangle,\langle0,0\rangle,\langle1,0\rangle,\langle0,0\rangle,\langle0,0\rangle,\langle0,0\rangle,\langle0,1\rangle\mid\\
&\langle0,0\rangle,\langle0,0\rangle,\langle0,0\rangle,\langle0,0\rangle,\langle0,0\rangle,\langle0,0\rangle,\langle0,0\rangle,\langle0,0\rangle,\langle0,0\rangle\mid\\
&\langle0,0\rangle,\langle0,0\rangle,\langle1,0\rangle,\langle0,0\rangle,\langle0,1\rangle,\langle0,0\rangle,\langle1,0\rangle,\langle0,0\rangle,\langle0,1\rangle\mid\\
&\langle0,0\rangle,\langle0,0\rangle,\langle0,0\rangle,\langle0,0\rangle,\langle0,0\rangle,\langle0,0\rangle,\langle0,0\rangle,\langle0,0\rangle,\langle0,0\rangle\mid\\
&\langle0,0\rangle,\langle0,0\rangle,\langle0,0\rangle,\langle0,0\rangle,\langle0,0\rangle,\langle0,0\rangle,\langle0,0\rangle,\langle0,0\rangle,\langle0,0\rangle\mid\\
&\langle0,0\rangle,\langle0,0\rangle,\langle0,0\rangle,\langle0,0\rangle,\langle0,0\rangle,\langle0,0\rangle,\langle0,0\rangle,\langle0,0\rangle,\langle0,0\rangle\mid\\
&\langle0,1\rangle,\langle0,0\rangle,\langle1,0\rangle,\langle0,0\rangle,\langle0,0\rangle,\langle0,0\rangle,\langle1,0\rangle,\langle0,0\rangle,\langle0,1\rangle\mid\\
&\langle0,0\rangle,\langle0,0\rangle,\langle0,0\rangle,\langle0,0\rangle,\langle0,0\rangle,\langle0,0\rangle,\langle0,0\rangle,\langle0,0\rangle,\langle0,0\rangle\mid\\
&\langle0,0\rangle,\langle0,0\rangle,\langle1,0\rangle,\langle0,0\rangle,\langle1,0\rangle,\langle0,0\rangle,\langle0,1\rangle,\langle0,0\rangle,\langle0,0\rangle).
\end{aligned}
\]

\[
\begin{aligned}
\alpha_{P'S'}={}&(\langle0,0\rangle,\langle0,0\rangle,\langle0,0\rangle,\langle0,0\rangle,\langle0,0\rangle,\langle0,0\rangle,\langle0,1\rangle,\langle0,0\rangle,\langle0,0\rangle\mid\\
&\langle0,0\rangle,\langle0,0\rangle,\langle0,0\rangle,\langle0,0\rangle,\langle0,0\rangle,\langle0,0\rangle,\langle0,0\rangle,\langle0,0\rangle,\langle0,0\rangle\mid\\
&\langle0,0\rangle,\langle0,0\rangle,\langle1,0\rangle,\langle0,0\rangle,\langle0,0\rangle,\langle0,0\rangle,\langle1,0\rangle,\langle0,0\rangle,\langle1,0\rangle\mid\\
&\langle0,0\rangle,\langle0,0\rangle,\langle0,0\rangle,\langle0,0\rangle,\langle0,0\rangle,\langle0,0\rangle,\langle0,0\rangle,\langle0,0\rangle,\langle0,0\rangle\mid\\
&\langle1,0\rangle,\langle0,0\rangle,\langle0,1\rangle,\langle0,0\rangle,\langle0,0\rangle,\langle0,0\rangle,\langle0,0\rangle,\langle0,0\rangle,\langle1,0\rangle\mid\\
&\langle0,0\rangle,\langle0,0\rangle,\langle0,0\rangle,\langle0,0\rangle,\langle0,0\rangle,\langle0,0\rangle,\langle0,0\rangle,\langle0,0\rangle,\langle0,0\rangle\mid\\
&\langle0,0\rangle,\langle0,0\rangle,\langle1,0\rangle,\langle0,0\rangle,\langle0,0\rangle,\langle0,0\rangle,\langle1,0\rangle,\langle0,0\rangle,\langle0,1\rangle\mid\\
&\langle0,0\rangle,\langle0,0\rangle,\langle0,0\rangle,\langle0,0\rangle,\langle0,0\rangle,\langle0,0\rangle,\langle0,0\rangle,\langle0,0\rangle,\langle0,0\rangle\mid\\
&\langle0,1\rangle,\langle0,0\rangle,\langle0,1\rangle,\langle0,0\rangle,\langle0,0\rangle,\langle0,0\rangle,\langle0,1\rangle,\langle0,0\rangle,\langle0,0\rangle).
\end{aligned}
\]

\[
\begin{aligned}
\alpha_{Q'P'}={}&(\langle0,1\rangle,\langle0,0\rangle,\langle0,0\rangle,\langle0,0\rangle,\langle1,0\rangle,\langle0,0\rangle,\langle0,0\rangle,\langle0,0\rangle,\langle0,0\rangle\mid\\
&\langle0,0\rangle,\langle0,0\rangle,\langle0,0\rangle,\langle0,0\rangle,\langle0,0\rangle,\langle0,0\rangle,\langle0,0\rangle,\langle0,0\rangle,\langle0,0\rangle\mid\\
&\langle0,1\rangle,\langle0,0\rangle,\langle1,0\rangle,\langle0,0\rangle,\langle0,1\rangle,\langle0,0\rangle,\langle1,0\rangle,\langle0,0\rangle,\langle0,0\rangle\mid\\
&\langle0,0\rangle,\langle0,0\rangle,\langle0,0\rangle,\langle0,0\rangle,\langle0,0\rangle,\langle0,0\rangle,\langle0,0\rangle,\langle0,0\rangle,\langle0,0\rangle\mid\\
&\langle0,0\rangle,\langle0,0\rangle,\langle0,0\rangle,\langle0,0\rangle,\langle0,0\rangle,\langle0,0\rangle,\langle0,0\rangle,\langle0,0\rangle,\langle0,0\rangle\mid\\
&\langle0,0\rangle,\langle0,0\rangle,\langle0,0\rangle,\langle0,0\rangle,\langle0,0\rangle,\langle0,0\rangle,\langle0,0\rangle,\langle0,0\rangle,\langle0,0\rangle\mid\\
&\langle0,1\rangle,\langle0,0\rangle,\langle1,0\rangle,\langle0,0\rangle,\langle0,0\rangle,\langle0,0\rangle,\langle1,0\rangle,\langle0,0\rangle,\langle0,1\rangle\mid\\
&\langle0,0\rangle,\langle0,0\rangle,\langle0,0\rangle,\langle0,0\rangle,\langle0,0\rangle,\langle0,0\rangle,\langle0,0\rangle,\langle0,0\rangle,\langle0,0\rangle\mid\\
&\langle0,0\rangle,\langle0,0\rangle,\langle0,1\rangle,\langle0,0\rangle,\langle1,0\rangle,\langle0,0\rangle,\langle1,0\rangle,\langle0,0\rangle,\langle0,0\rangle).
\end{aligned}
\]

\[
\begin{aligned}
\alpha_{Q'R'}={}&(\langle0,1\rangle,\langle0,0\rangle,\langle0,1\rangle,\langle0,0\rangle,\langle0,0\rangle,\langle0,0\rangle,\langle0,1\rangle,\langle0,0\rangle,\langle0,0\rangle\mid\\
&\langle0,0\rangle,\langle0,0\rangle,\langle0,0\rangle,\langle0,0\rangle,\langle0,0\rangle,\langle0,0\rangle,\langle0,0\rangle,\langle0,0\rangle,\langle0,0\rangle\mid\\
&\langle0,0\rangle,\langle0,0\rangle,\langle1,0\rangle,\langle0,0\rangle,\langle0,0\rangle,\langle0,0\rangle,\langle1,0\rangle,\langle0,0\rangle,\langle0,1\rangle\mid\\
&\langle0,0\rangle,\langle0,0\rangle,\langle0,0\rangle,\langle0,0\rangle,\langle0,0\rangle,\langle0,0\rangle,\langle0,0\rangle,\langle0,0\rangle,\langle0,0\rangle\mid\\
&\langle1,0\rangle,\langle0,0\rangle,\langle0,1\rangle,\langle0,0\rangle,\langle0,0\rangle,\langle0,0\rangle,\langle0,0\rangle,\langle0,0\rangle,\langle1,0\rangle\mid\\
&\langle0,0\rangle,\langle0,0\rangle,\langle0,0\rangle,\langle0,0\rangle,\langle0,0\rangle,\langle0,0\rangle,\langle0,0\rangle,\langle0,0\rangle,\langle0,0\rangle\mid\\
&\langle0,0\rangle,\langle0,0\rangle,\langle1,0\rangle,\langle0,0\rangle,\langle0,0\rangle,\langle0,0\rangle,\langle1,0\rangle,\langle0,0\rangle,\langle1,0\rangle\mid\\
&\langle0,0\rangle,\langle0,0\rangle,\langle0,0\rangle,\langle0,0\rangle,\langle0,0\rangle,\langle0,0\rangle,\langle0,0\rangle,\langle0,0\rangle,\langle0,0\rangle\mid\\
&\langle0,0\rangle,\langle0,0\rangle,\langle0,0\rangle,\langle0,0\rangle,\langle0,0\rangle,\langle0,0\rangle,\langle0,1\rangle,\langle0,0\rangle,\langle0,0\rangle).
\end{aligned}
\]

\[
\begin{aligned}
\alpha_{R'Q'}={}&(\langle0,0\rangle,\langle0,0\rangle,\langle0,1\rangle,\langle0,0\rangle,\langle0,0\rangle,\langle0,0\rangle,\langle0,1\rangle,\langle0,0\rangle,\langle0,1\rangle\mid\\
&\langle0,0\rangle,\langle0,0\rangle,\langle0,0\rangle,\langle0,0\rangle,\langle0,0\rangle,\langle0,0\rangle,\langle0,0\rangle,\langle0,0\rangle,\langle0,0\rangle\mid\\
&\langle0,1\rangle,\langle0,0\rangle,\langle1,0\rangle,\langle0,0\rangle,\langle0,0\rangle,\langle0,0\rangle,\langle1,0\rangle,\langle0,0\rangle,\langle0,0\rangle\mid\\
&\langle0,0\rangle,\langle0,0\rangle,\langle0,0\rangle,\langle0,0\rangle,\langle0,0\rangle,\langle0,0\rangle,\langle0,0\rangle,\langle0,0\rangle,\langle0,0\rangle\mid\\
&\langle1,0\rangle,\langle0,0\rangle,\langle0,0\rangle,\langle0,0\rangle,\langle0,0\rangle,\langle0,0\rangle,\langle0,1\rangle,\langle0,0\rangle,\langle1,0\rangle\mid\\
&\langle0,0\rangle,\langle0,0\rangle,\langle0,0\rangle,\langle0,0\rangle,\langle0,0\rangle,\langle0,0\rangle,\langle0,0\rangle,\langle0,0\rangle,\langle0,0\rangle\mid\\
&\langle1,0\rangle,\langle0,0\rangle,\langle1,0\rangle,\langle0,0\rangle,\langle0,0\rangle,\langle0,0\rangle,\langle1,0\rangle,\langle0,0\rangle,\langle0,0\rangle\mid\\
&\langle0,0\rangle,\langle0,0\rangle,\langle0,0\rangle,\langle0,0\rangle,\langle0,0\rangle,\langle0,0\rangle,\langle0,0\rangle,\langle0,0\rangle,\langle0,0\rangle\mid\\
&\langle0,0\rangle,\langle0,0\rangle,\langle0,1\rangle,\langle0,0\rangle,\langle0,0\rangle,\langle0,0\rangle,\langle0,0\rangle,\langle0,0\rangle,\langle0,0\rangle).
\end{aligned}
\]

\[
\begin{aligned}
\alpha_{R'S'}={}&(\langle0,0\rangle,\langle0,0\rangle,\langle0,1\rangle,\langle0,0\rangle,\langle1,0\rangle,\langle0,0\rangle,\langle1,0\rangle,\langle0,0\rangle,\langle0,0\rangle\mid\\
&\langle0,0\rangle,\langle0,0\rangle,\langle0,0\rangle,\langle0,0\rangle,\langle0,0\rangle,\langle0,0\rangle,\langle0,0\rangle,\langle0,0\rangle,\langle0,0\rangle\mid\\
&\langle0,1\rangle,\langle0,0\rangle,\langle1,0\rangle,\langle0,0\rangle,\langle0,0\rangle,\langle0,0\rangle,\langle1,0\rangle,\langle0,0\rangle,\langle0,1\rangle\mid\\
&\langle0,0\rangle,\langle0,0\rangle,\langle0,0\rangle,\langle0,0\rangle,\langle0,0\rangle,\langle0,0\rangle,\langle0,0\rangle,\langle0,0\rangle,\langle0,0\rangle\mid\\
&\langle0,0\rangle,\langle0,0\rangle,\langle0,0\rangle,\langle0,0\rangle,\langle0,0\rangle,\langle0,0\rangle,\langle0,0\rangle,\langle0,0\rangle,\langle0,0\rangle\mid\\
&\langle0,0\rangle,\langle0,0\rangle,\langle0,0\rangle,\langle0,0\rangle,\langle0,0\rangle,\langle0,0\rangle,\langle0,0\rangle,\langle0,0\rangle,\langle0,0\rangle\mid\\
&\langle0,1\rangle,\langle0,0\rangle,\langle1,0\rangle,\langle0,0\rangle,\langle0,1\rangle,\langle0,0\rangle,\langle1,0\rangle,\langle0,0\rangle,\langle0,0\rangle\mid\\
&\langle0,0\rangle,\langle0,0\rangle,\langle0,0\rangle,\langle0,0\rangle,\langle0,0\rangle,\langle0,0\rangle,\langle0,0\rangle,\langle0,0\rangle,\langle0,0\rangle\mid\\
&\langle0,1\rangle,\langle0,0\rangle,\langle0,0\rangle,\langle0,0\rangle,\langle1,0\rangle,\langle0,0\rangle,\langle0,0\rangle,\langle0,0\rangle,\langle0,0\rangle).
\end{aligned}
\]

\[
\begin{aligned}
\alpha_{S'P'}={}&(\langle0,0\rangle,\langle0,0\rangle,\langle0,1\rangle,\langle0,0\rangle,\langle0,0\rangle,\langle0,0\rangle,\langle0,0\rangle,\langle0,0\rangle,\langle0,0\rangle\mid\\
&\langle0,0\rangle,\langle0,0\rangle,\langle0,0\rangle,\langle0,0\rangle,\langle0,0\rangle,\langle0,0\rangle,\langle0,0\rangle,\langle0,0\rangle,\langle0,0\rangle\mid\\
&\langle1,0\rangle,\langle0,0\rangle,\langle1,0\rangle,\langle0,0\rangle,\langle0,0\rangle,\langle0,0\rangle,\langle1,0\rangle,\langle0,0\rangle,\langle0,0\rangle\mid\\
&\langle0,0\rangle,\langle0,0\rangle,\langle0,0\rangle,\langle0,0\rangle,\langle0,0\rangle,\langle0,0\rangle,\langle0,0\rangle,\langle0,0\rangle,\langle0,0\rangle\mid\\
&\langle1,0\rangle,\langle0,0\rangle,\langle0,0\rangle,\langle0,0\rangle,\langle0,0\rangle,\langle0,0\rangle,\langle0,1\rangle,\langle0,0\rangle,\langle1,0\rangle\mid\\
&\langle0,0\rangle,\langle0,0\rangle,\langle0,0\rangle,\langle0,0\rangle,\langle0,0\rangle,\langle0,0\rangle,\langle0,0\rangle,\langle0,0\rangle,\langle0,0\rangle\mid\\
&\langle0,1\rangle,\langle0,0\rangle,\langle1,0\rangle,\langle0,0\rangle,\langle0,0\rangle,\langle0,0\rangle,\langle1,0\rangle,\langle0,0\rangle,\langle0,0\rangle\mid\\
&\langle0,0\rangle,\langle0,0\rangle,\langle0,0\rangle,\langle0,0\rangle,\langle0,0\rangle,\langle0,0\rangle,\langle0,0\rangle,\langle0,0\rangle,\langle0,0\rangle\mid\\
&\langle0,0\rangle,\langle0,0\rangle,\langle0,1\rangle,\langle0,0\rangle,\langle0,0\rangle,\langle0,0\rangle,\langle0,1\rangle,\langle0,0\rangle,\langle0,1\rangle).
\end{aligned}
\]

\[
\begin{aligned}
\alpha_{S'R'}={}&(\langle0,0\rangle,\langle0,0\rangle,\langle1,0\rangle,\langle0,0\rangle,\langle1,0\rangle,\langle0,0\rangle,\langle0,1\rangle,\langle0,0\rangle,\langle0,0\rangle\mid\\
&\langle0,0\rangle,\langle0,0\rangle,\langle0,0\rangle,\langle0,0\rangle,\langle0,0\rangle,\langle0,0\rangle,\langle0,0\rangle,\langle0,0\rangle,\langle0,0\rangle\mid\\
&\langle0,1\rangle,\langle0,0\rangle,\langle1,0\rangle,\langle0,0\rangle,\langle0,0\rangle,\langle0,0\rangle,\langle1,0\rangle,\langle0,0\rangle,\langle0,1\rangle\mid\\
&\langle0,0\rangle,\langle0,0\rangle,\langle0,0\rangle,\langle0,0\rangle,\langle0,0\rangle,\langle0,0\rangle,\langle0,0\rangle,\langle0,0\rangle,\langle0,0\rangle\mid\\
&\langle0,0\rangle,\langle0,0\rangle,\langle0,0\rangle,\langle0,0\rangle,\langle0,0\rangle,\langle0,0\rangle,\langle0,0\rangle,\langle0,0\rangle,\langle0,0\rangle\mid\\
&\langle0,0\rangle,\langle0,0\rangle,\langle0,0\rangle,\langle0,0\rangle,\langle0,0\rangle,\langle0,0\rangle,\langle0,0\rangle,\langle0,0\rangle,\langle0,0\rangle\mid\\
&\langle0,0\rangle,\langle0,0\rangle,\langle1,0\rangle,\langle0,0\rangle,\langle0,1\rangle,\langle0,0\rangle,\langle1,0\rangle,\langle0,0\rangle,\langle0,1\rangle\mid\\
&\langle0,0\rangle,\langle0,0\rangle,\langle0,0\rangle,\langle0,0\rangle,\langle0,0\rangle,\langle0,0\rangle,\langle0,0\rangle,\langle0,0\rangle,\langle0,0\rangle\mid\\
&\langle0,0\rangle,\langle0,0\rangle,\langle0,0\rangle,\langle0,0\rangle,\langle1,0\rangle,\langle0,0\rangle,\langle0,0\rangle,\langle0,0\rangle,\langle0,1\rangle).
\end{aligned}
\]
\endgroup

The canonical string at each corner is
\[
\begin{aligned}
\alpha_{P'}&=\min\{\alpha_{P'Q'},\alpha_{P'S'}\}=\alpha_{P'S'},\\
\alpha_{Q'}&=\min\{\alpha_{Q'P'},\alpha_{Q'R'}\}=\alpha_{Q'P'},\\
\alpha_{R'}&=\min\{\alpha_{R'Q'},\alpha_{R'S'}\}=\alpha_{R'Q'},\\
\alpha_{S'}&=\min\{\alpha_{S'P'},\alpha_{S'R'}\}=\alpha_{S'R'}.
\end{aligned}
\]
Moreover,
\[
\alpha_{P'S'}=
\min\left\{
\alpha_{P'Q'},\alpha_{P'S'},\alpha_{Q'P'},\alpha_{Q'R'},
\alpha_{R'Q'},\alpha_{R'S'},\alpha_{S'P'},\alpha_{S'R'}
\right\}.
\]
Therefore,
\[
\boxed{P'\text{ is the unique anchor corner, with direction }P'S'.}
\]

Altogether,
\[
\boxed{
\begin{gathered}
\mathcal W(t_0)=\{m_2,m_3\},\qquad
\mathcal W_c(t_0)=\{m_2\},\\
\text{pivot corner}=S,\qquad
\text{anchor corner}=P'.
\end{gathered}
}
\]

\section*{Proof of the Crash-Induced Symmetry Impossibility}
\label{app:crash-induced-symmetry-proofs}

\noindent\textbf{Theorem~\ref{thm:weak-multiplicity-impossible} (Impossibility under weak multiplicity detection).}
There is no deterministic distributed algorithm that solves optimal Weber-meeting-node gathering for crash-faulty robots up to $n-2$ operating under \textsc{FSYNC} with rigid movement, if robots are equipped only with \emph{weak} multiplicity detection.

\begin{proof}
Assume that a deterministic distributed algorithm \(\mathcal{A}\)
solves optimal gathering under \textsc{FSYNC} activation and rigid
movements using only global weak multiplicity detection.

Let $m_1=(0,0)\ \text{and}\ m_2=(8,0),$ and consider the multiplicity-free initial configuration
\[
\begin{aligned}
r_1&=(2,0),   & r_2&=(3,0),   & r_3&=(4,0),\\
r_4&=(5,0),   & r_5&=(6,0),  & r_6&=(4,-2),\\
r_7&=(3,1),   & r_8&=(5,1),   & r_9&=(3,-1),\\
r_{10}&=(4,-1).
\end{aligned}
\]
Let \(r_1,r_2,r_4\), and \(r_6\) be crash-faulty initially, hence $4\le n-2=10-2=8$. The aggregate costs are $c_{t_0}(m_1)=45<47=c_{t_0}(m_2),$ so \(m_1\) is the unique initial \textsc{Weber Meeting Node}.

Consider the first \textsc{FSYNC} round $C_1$, in which the non-faulty
robots make rigid one-hop moves toward $m_1$ and reach the following
nodes:
\[
\begin{aligned}
r_3,r_7,r_9&\longrightarrow p_2=(3,0),\\
r_{10}&\longrightarrow q=(4,0),\\
r_5,r_8&\longrightarrow p_4=(5,0).
\end{aligned}
\]
Let $p_1=(2,0)\ \text{and}\ p_6=(4,-2).$ At the end of $C_1$, $p_2$ and $p_4$ are multiplicity nodes containing exactly four and three robots, respectively. Accordingly, with respect to the ordered nodes $(p_1,p_2,q,p_6,p_4)$, the exact multiplicity profile is $\Lambda=(1,4,1,1,3)$ and total number of robots is $10=1+4+1+1+3$.

\begin{figure}[ht!]
    \centering

    \begin{subfigure}[t]{0.47\linewidth}
        \centering
        \includegraphics[width=\linewidth]{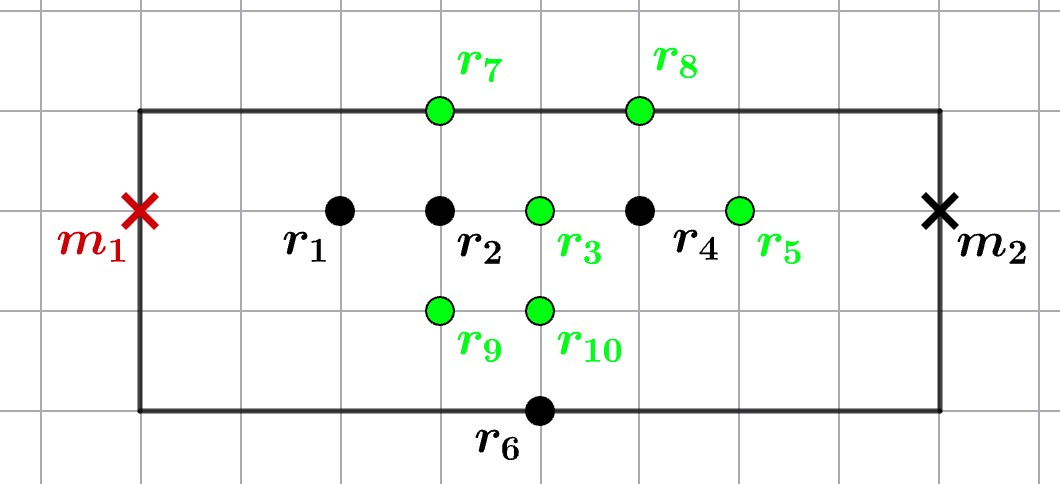}
        \caption{}
   
    \end{subfigure}
    \hfill
    \begin{subfigure}[t]{0.47\linewidth}
        \centering
        \includegraphics[width=\linewidth]{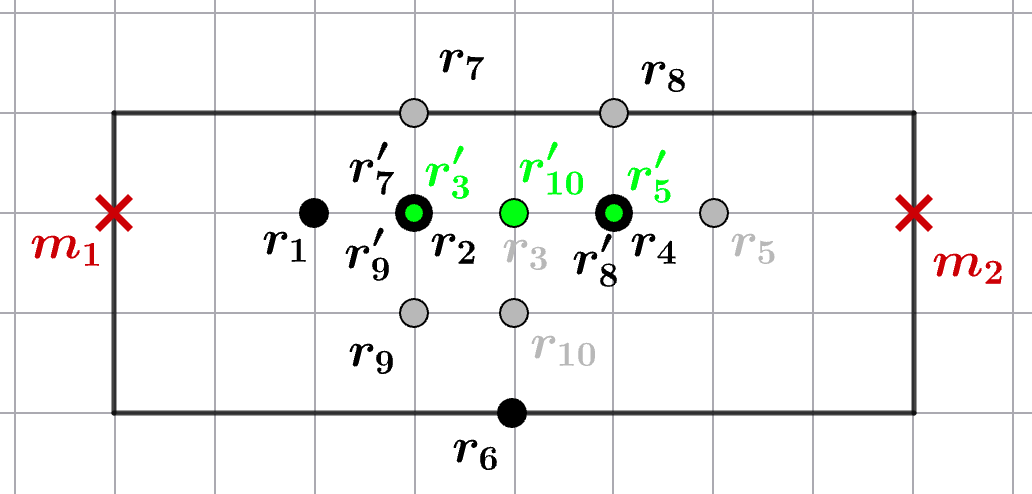}
        \caption{}
       
    \end{subfigure}\textbf{}
\caption{}
    \label{fig:strong-multiplicity}
\end{figure}

Immediately before the next round \(C_2\), let \(r_7\), \(r_8\) and  \(r_9\)
also get crashed. Thus, only \(r_3,r_5\), and \(r_{10}\) remain non-faulty,
and the total number of crash-faulty robots is $7\leq n-2=8.$ Since the robots are equipped only with global weak multiplicity detection, they observe $p_2$ and $p_4$ as multiplicity nodes but cannot determine their exact multiplicities. Their weak multiplicity profile, with respect to the ordered nodes $(p_1,p_2,q,p_6,p_4)$, is therefore $(1,\mathsf{M},1,1,\mathsf{M})$. Even if the robots know that $n=10$, the three singly occupied nodes $p_1$, $q$, and $p_6$ allow them to infer only that
$\lambda(p_2)+\lambda(p_4)=7$, where both $\lambda(p_2)$ and $\lambda(p_4)$ are at least two. Consequently, both distributions $(2,5)$ and $(5,2)$ are consistent with the same weak snapshot. Suppose that the computations of $r_5$ and $r_{10}$ are based on the compatible exact multiplicity profile $\Lambda_{r_5}=\Lambda_{r_{10}}=(1,2,1,1,5)$, whereas the computation of $r_3$ is based on the compatible profile
$\Lambda_{r_3}=(1,5,1,1,2)$. In both cases, the entries of the profile
sum to ten: $1+2+1+1+5=1+5+1+1+2=10$.

Consider an execution $E$ of algorithm $\mathcal A$ and the
\textsc{FSYNC} round $C_2$ occurring at time $t'>t_0$. Since all
non-faulty robots are activated in this round, they perform their
\textit{Look} and \textit{Compute} phases simultaneously. The computations of $r_5$ and $r_{10}$ are based on the compatible
multiplicity profile $\Lambda_{r_5}=\Lambda_{r_{10}}=(1,2,1,1,5)$. For this profile, the current Weber costs are $\bigl(c_{t'}(m_1),c_{t'}(m_2)\bigr)=(43,41)$. Hence, both $r_5$ and $r_{10}$ identify $m_2$ as the unique \textsc{Weber Meeting Node} during their \textit{Compute} phases. In contrast, the computation of the non-faulty robot $r_3$ is based on the compatible multiplicity profile
$\Lambda_{r_3}=(1,5,1,1,2)$. In this case, the current Weber costs are
$\bigl(c_{t'}(m_1),c_{t'}(m_2)\bigr)=(37,47)$. Therefore, $r_3$
identifies $m_1$ as the unique \textsc{Weber Meeting Node}.

During the \textit{Move} phase of the same round $C_2$, robot $r_3$, having identified $m_1$ as the unique \textsc{Weber Meeting Node}, makes at least one rigid hop toward $m_1$. In contrast, robots $r_5$ and $r_{10}$, having identified $m_2$ as the unique \textsc{Weber Meeting Node}, each make at least one rigid hop toward $m_2$. Thus, their first hops in round $C_2$ are
\[
r_3:(3,0)\longrightarrow(2,0),\qquad
r_{10}:(4,0)\longrightarrow(5,0),\qquad
r_5:(5,0)\longrightarrow(6,0).
\]

At time $t_0$, $m_1$ was the unique \textsc{Weber Meeting Node} of the initial configuration. During round $C_1$, all non-faulty robots moved along shortest paths toward $m_1$. However, in round $C_2$, robots $r_5$ and $r_{10}$ move rightward toward $m_2$, thereby increasing their distances from $m_1$ by one. These movements do not lie on any shortest path from their initial positions to $m_1$. In fact, both robots first move one hop toward $m_1$ in round $C_1$ and then move one hop in the opposite direction in round $C_2$. This reversal adds at least two grid edges to the total distance traveled by each robot compared with a shortest path to $m_1$. Therefore, the min-sum optimality condition is violated as soon as $r_5$ and $r_{10}$ make their first rightward hops in round $C_2$.

Thus, the indistinguishability of the exact multiplicity profiles under global weak multiplicity detection causes the non-faulty robots to move toward different meeting nodes, leading to a violation of the min-sum optimality condition. Consequently, algorithm $\mathcal A$ cannot solve optimal gathering in both indistinguishable executions. This completes the proof.
\hfill\qed
\end{proof}

\begin{lemma}[Unsolvability when all robots on the symmetry line crash]
\label{lem:faulty-axis-unsolvable}
Suppose that the initial entire configuration $\mathcal I(t_0)$ is symmetric with respect to a unique line $\mathcal L$, and that $\mathcal M\cap\mathcal L=\varnothing$. If all robots initially located on $\mathcal {L} $ and crash before performing any move, while at least two robots outside $\mathcal {L} $ remain non-faulty, then gathering all non-faulty robots at a meeting node is unsolvable. Consequently, gathering at a \textsc{Weber Meeting Node} is also unsolvable.
\end{lemma}

\begin{proof}
Consider an admissible execution in which all robots on $\mathcal L$ crash
before moving and the non-faulty robots outside $\mathcal L$ are activated
symmetrically. Since the robots are anonymous and deterministic, the
reflectional symmetry with respect to $\mathcal L$ is preserved throughout
the execution. Therefore, any common gathering node must lie on
$\mathcal L$. However, $\mathcal M\cap\mathcal L=\varnothing$, so no meeting
node lies on $\mathcal L$. Hence, gathering at a meeting node is unsolvable.
Since every \textsc{Weber Meeting Node} is a meeting node, optimal gathering
at a \textsc{Weber Meeting Node} is also unsolvable.
\hfill\qed
\end{proof}

\begin{lemma}[Unsolvability when the robot at the center of rotation crashes]
\label{lem:faulty-center-unsolvable}
Suppose that the initial entire configuration $\mathcal I(t_0)$ is rotationally
symmetric about a center $c$, and that $c\notin\mathcal M$. If the robot
initially located at $c$ crashes before performing any move, while at least
two robots outside $c$ remain non-faulty, then gathering all non-faulty
robots at a meeting node is unsolvable. Consequently, gathering at a
\textsc{Weber Meeting Node} is also unsolvable.
\end{lemma}

\begin{proof}
Consider an admissible execution in which the robot at $c$ crashes before
moving, while the non-faulty robots outside $c$ are activated symmetrically.
Since the robots are anonymous and deterministic, the rotational symmetry
about $c$ is preserved throughout the execution. Hence, any common gathering
node must be fixed by the rotation and therefore must coincide with $c$.
However, $c\notin\mathcal M$, so $c$ is not a meeting node. Thus, gathering
at a meeting node is unsolvable. Since every \textsc{Weber Meeting Node} is
a meeting node, gathering at a \textsc{Weber Meeting Node} is also
unsolvable.
\hfill\qed
\end{proof}

\begin{lemma}[Unsolvability when robots on two perpendicular symmetry lines crash]
\label{lem:faulty-perpendicular-lines-unsolvable}
Suppose that the initial entire configuration $\mathcal I(t_0)$ is rotationally
symmetric about a grid node $c$. Let $\mathcal L_1$ and $\mathcal L_2$ be two
perpendicular grid lines passing through $c$, whose union is invariant under
the same rotation. If every robot initially located on $\mathcal L_1\cup\mathcal L_2$ crashes before performing any move, while at least two robots outside $\mathcal L_1\cup\mathcal L_2$ remain non-faulty, then gathering all non-faulty robots at a meeting node is unsolvable. Consequently, gathering at a \textsc{Weber Meeting Node} is also unsolvable.
\end{lemma}

\begin{proof}
The crashed robots remain stationary, and their configuration on
$\mathcal L_1\cup\mathcal L_2$ preserves the rotational symmetry about $c$.
The adversary can activate the non-faulty robots outside these lines in
symmetric groups, thereby preserving the same rotational symmetry throughout
the execution. Hence, any common gathering node must be fixed by the rotation
and must therefore coincide with $c$. Since $c\notin\mathcal M$, gathering at
a meeting node is impossible. Consequently, gathering at a
\textsc{Weber Meeting Node} is also impossible.
\hfill\qed
\end{proof}

In the remainder of the paper, we assume that whenever an initial configuration admits a unique line of symmetry $\mathcal L$, the line $\mathcal L$ passes through grid nodes. Similarly, whenever an initial configuration admits a nontrivial rotational symmetry, its center of rotation is a grid node. Let $\mathcal I(t_0)$ denote the set of all initial configurations, and let $\mathbb U\subseteq\mathbb I$ denote the set of initial configurations from which gathering is impossible, as established by Lemmas~\ref{lem:faulty-axis-unsolvable}--\ref{lem:faulty-perpendicular-lines-unsolvable}. Thus, $\mathcal I(t_0)\setminus\mathbb U$ denotes the set of initial configurations considered for the proposed gathering algorithm.

\begin{figure}[ht!]
\centering
\includegraphics[width=0.4\linewidth]{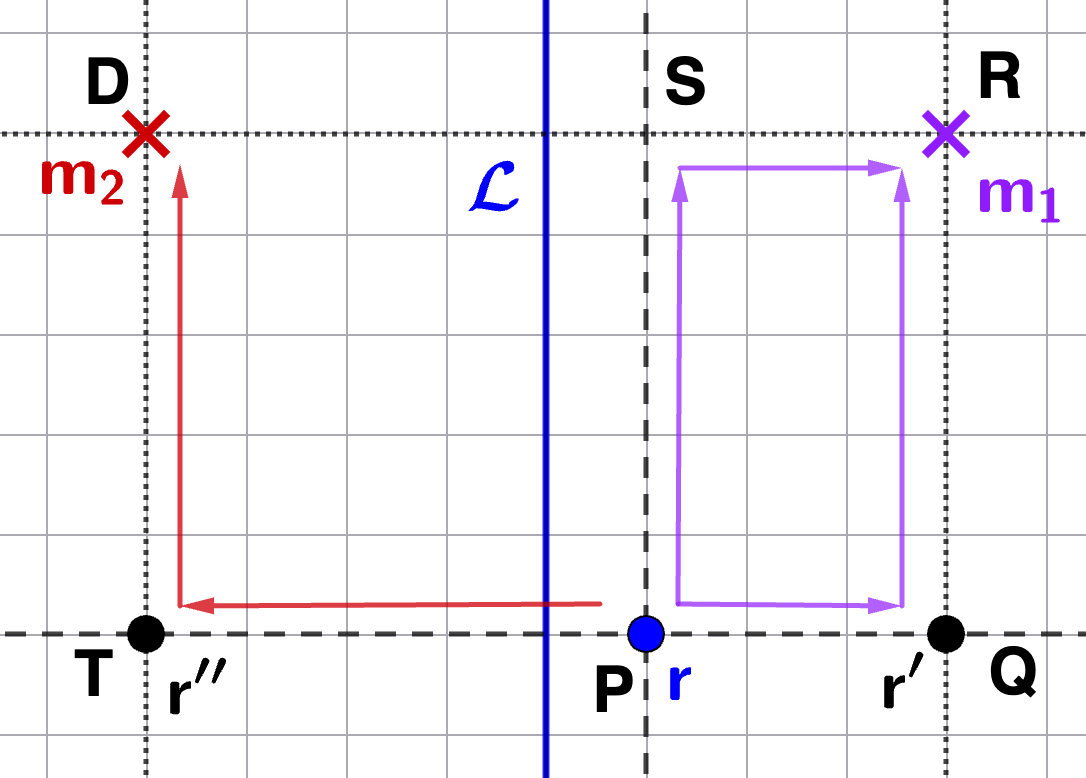}
\caption{Target-dependent restricted shortest paths. Robot $r$ may follow
one of the purple paths through $S$ or $r'$ toward $m_1$. The red path
through $r''$ leads toward the reflected candidate $m_2$.}
\label{fig:restricted-path}
\end{figure}

\section*{Full Correctness Proofs}
\label{app:correctness-proofs}

For readability, each result stated in the main correctness section is
restated before its complete proof.


\noindent\textbf{Lemma~\ref{lem:region-unique}
(Uniqueness of the target region).}
For every configuration in $\mathcal I_{321}$, $\mathcal I_{322}$, or
$\mathcal I_{323}$, Phase~I-A uniquely selects one of $\mathcal B_1$ and
$\mathcal B_2$. For every configuration in
$\mathcal I_{4211}$, $\mathcal I_{42121}$,
$\mathcal I_{42122}^{(1)}$, $\mathcal I_{42122}^{(2)}$,
$\mathcal I_{42122}^{(3)}$, $\mathcal I_{4221}$, or
$\mathcal I_{42221}$, it uniquely selects one of
$\mathcal Q_1,\ldots,\mathcal Q_4$. 



\begin{proof}
Consider first the reflectional classes. Phase~I-A initially compares
$\omega_1$ and $\omega_2$. If $\omega_i>\omega_j$, then
$\mathcal B_i$ is uniquely selected. If $\omega_1=\omega_2$, the two
half-planes are compared by their robot occupancies. Thus,
$\nu_i>\nu_j$ uniquely determines $\mathcal B^\dagger=\mathcal B_i$.

Suppose that both comparisons are tied. In $\mathcal I_{321}$, the entire configuration is asymmetric and therefore has a unique anchor corner.
Exactly one of the two half-planes does not contain this corner, and hence the prescribed rule selects it uniquely.

For $\mathcal I_{322}$, if
$\mathcal{SER}_{\mathcal I(t)}=\mathcal{SER}_{\mathcal M}$, the
strings $\alpha_{PS}$ and $\alpha_{PQ}$ are compared. These strings cannot be equal, since their equality would imply identical
meeting-node indicators and robot multiplicities at every pair of nodes
related by reflection across $\mathcal L$.
Consequently, the entire configuration would be reflectionally symmetric. Therefore, exactly one of $\mathcal B_{PS}$ and $\mathcal B_{PQ}$ is
selected. If $
\mathcal{SER}_{\mathcal M}\subsetneq\mathcal{SER}_{\mathcal I(t)}
$, the unequal
values of $\eta_1$ and $\eta_2$ determine the target half-plane uniquely.
When $\eta_1=\eta_2$, equality of the two prescribed canonical strings from
the anchor corner would again imply reflectional symmetry of the entire
configuration. Hence, their lexicographic comparison is strict.

The same argument applies to $\mathcal I_{323}$. If the two enclosing
rectangles coincide, equality of the two pivot-corner strings compared by
the rule would imply reflectional symmetry. Otherwise, an unequal grid-node
counts determine the target half-plane, and a remaining tie is resolved by
the two canonical strings associated with the unique anchor corner. Since
the configuration is asymmetric, these strings are distinct. Thus, each
reflectional class selects exactly one of $\mathcal B_1$ and
$\mathcal B_2$.

For $\mathcal I_{4211}$, the uniqueness of $\mathcal Q^\dagger$ follows
from the unique maximum candidate count. For $\mathcal I_{42121}$ and
$\mathcal I_{4221}$, it follows from the unique maximum occupancy among the
specified quadrants. In each subclass of $\mathcal I_{42122}$ and in
$\mathcal I_{42221}$, the remaining tied quadrants are compared by the
prescribed $\alpha$- or $\beta$-strings. If two distinct competing corners had the same final string, then the
entire configuration would admit the corresponding rotational symmetry,
which contradicts the assumption that the configuration is asymmetric. Hence, the final comparison is strict and
exactly one quadrant is selected.
\end{proof}

\noindent\textbf{Lemma~\ref{lem:node-unique}
(Uniqueness of the target node).}
For every solvable configuration, Phase~I determines exactly one target
\textsc{Weber Meeting Node} $m^\dagger\in\mathcal W(t)$ from the current
snapshot.

\begin{proof}
In $\mathcal I_1$, the single node of $\mathcal W(t)$ is selected. In
$\mathcal I_2$, the asymmetry of $\mathcal M$ determines a unique canonical
traversal and hence a unique outcome of the prescribed candidate-selection
rule. For $\mathcal I_{31}$, the rule for the reflection axis, together with
the anchor-corner rule when two candidates occur, selects exactly one node.
In $\mathcal I_{41}$, the algorithm selects the rotation center $c$ as a
\textsc{Weber Meeting Node}; however, this does not imply that $c$ is the
unique \textsc{Weber Meeting Node}.

For $\mathcal I_{321}$, $\mathcal I_{322}$, and
$\mathcal I_{323}$, Lemma~\ref{lem:region-unique} first provides a unique
target half-plane $\mathcal B^\dagger$. The set
$\mathcal W_c(t)\cap\mathcal B^\dagger$ is nonempty. If the set contains only one element, that element is selected.
If it contains two candidates, the fixed traversal direction of
$\mathcal B^\dagger$ uniquely determines the first candidate encountered.

For $\mathcal I_{4211}$, $\mathcal I_{42121}$, and the three subclasses of
$\mathcal I_{42122}$, Lemma~\ref{lem:region-unique} gives a unique
$\mathcal Q^\dagger$. Candidates in $\mathcal W^\dagger$ are first compared according to their
distance from $p^\dagger$. If a tie remains, their distance from
$\mathcal L_1^{++}$ is used, and the conditions of the class ensure that
the selected candidate is unique.
For $\mathcal I_{4221}$ and $\mathcal I_{42221}$, the selected quadrant and
asymmetry determine a unique canonical traversal. Since
$\overline{\mathcal W}^\dagger\neq\varnothing$, the first node of this set
encountered by that traversal is well defined and unique. These rules depend only on coordinate-invariant distances and canonical
strings. Therefore, all robots observing the same snapshot select the same
meeting node $m^\dagger$.
\end{proof}


\noindent\textbf{Lemma~\ref{lem:halfplane-preservation}
(Invariance of the selected half-plane).}
Let $\mathcal B^\dagger$ be the target half-plane selected from a
configuration in $\mathcal I_{321}$, $\mathcal I_{322}$, or
$\mathcal I_{323}$. If every subsequent move follows
\textsc{MoveToDestination()}, every configuration reached during Phase~II
selects the same half-plane $\mathcal B^\dagger$.

\begin{proof}
Let $\mathcal B$ be the half-plane different from
$\mathcal B^\dagger$. By \textsc{MoveToDestination()}, a robot in
$\mathcal B^\dagger$ never enters $\mathcal B$. A robot in $\mathcal B$ may cross $\mathcal L$ only to enter
$\mathcal B^\dagger$, and once it crosses, it cannot return.
Similarly, a robot on $\mathcal L$ may leave the line only toward
$\mathcal B^\dagger$. Hence, the occupancy of $\mathcal B^\dagger$ cannot
decrease.

Consider the first criterion of Phase~I-A that distinguishes the two
half-planes. If the comparison is based on occupancy, the monotonicity established above
preserves the strict inequality. If it is based on one of the prescribed
geometric sums, every allowed move preserves the inequality favoring
$\mathcal B^\dagger$. Finally, if the selection is based on canonical strings, the restricted
paths preserve the first differing entry that favors $\mathcal B^\dagger$,
or create an earlier differing entry with the same preference. Moreover,
no robot moves to an earlier position in the traversal of $\mathcal B$. Thus, the first strict comparison cannot reverse after a
completed move. Induction over completed moves proves that every subsequent
snapshot selects $\mathcal B^\dagger$.
\end{proof}


\noindent\textbf{Lemma~\ref{lem:quadrant-preservation}
(Invariance of the selected quadrant).}
Let $\mathcal Q^\dagger$ be the target quadrant selected at the beginning of
Phase~II. Every subsequent configuration reached during Phase~II selects the
same quadrant $\mathcal Q^\dagger$.

\begin{proof}
Under the prescribed paths, a robot in $\mathcal Q^\dagger$ remains there.
A robot in a quadrant adjacent to $\mathcal Q^\dagger$ can enter it only
across their common boundary, whereas a robot in the opposite quadrant
must pass through $c$. Once a robot enters $\mathcal Q^\dagger$, it remains
there. Robots on the bounding lines are
treated according to the same convention in every snapshot. Consequently,
no move transfers a robot from $\mathcal Q^\dagger$ to a non-selected open
quadrant.

Fix $\mathcal Q\neq\mathcal Q^\dagger$ and consider the first criterion that
distinguishes them. For a comparison based on the distribution of candidates,
Lemma~\ref{wbn} ensures that no new \textsc{Weber Meeting Node} is created,
while $m^\dagger$ remains associated with $\mathcal Q^\dagger$.
Hence, the class-specific pairing of candidates cannot make
$\mathcal Q$ preferable to $\mathcal Q^\dagger$. For an occupancy comparison,
the prescribed transfers do not decrease the occupancy of
$\mathcal Q^\dagger$ or increase that of $\mathcal Q$. If the first
distinction is an $\alpha$- or $\beta$-string, the restricted paths preserve
the first unequal entry favoring $\mathcal Q^\dagger$ or replace it with an
earlier entry having the same preference. Hence, the first strict criterion
cannot reverse. This holds for every non-selected quadrant and after every
completed move, so $\mathcal Q^\dagger$ remains selected.
\end{proof}


\noindent\textbf{Lemma~\ref{lem:target-invariant}
(Invariance of the selected target).}
Let $m^\dagger$ be selected from $\mathcal I(t)$. If every move completed
after time $t$ follows \textsc{MoveToDestination()} toward $m^\dagger$, then
every subsequent snapshot selects the same target $m^\dagger$.

\begin{proof}
By Lemmas~\ref{lem:halfplane-preservation} and
\ref{lem:quadrant-preservation}, the selected region does not change.
These lemmas also preserve every anchor-based comparison used in selecting
that region. Furthermore, the symmetry axes, rotation center, pivot corners,
boundary half-lines, and meeting-node traversal directions depend only on the
fixed meeting-node set $\mathcal M$.

Every completed move is a shortest-path hop toward $m^\dagger$.
Lemma~\ref{wbn} therefore implies, for every subsequently reached
$\mathcal I(t')$,
$m^\dagger\in\mathcal W(t')$ and
$\mathcal W(t')\subseteq\mathcal W(t)$. Thus, $m^\dagger$ remains eligible
and no new candidate can appear. Under a distance-based rule, deleting other
candidates cannot create a candidate preceding $m^\dagger$ in the fixed
distance order. Under a traversal-based rule, no earlier node can become a
new candidate. The applicable rule consequently selects $m^\dagger$ in every
subsequent snapshot.
\end{proof}


\noindent\textbf{Lemma~\ref{lem:crash-noninterference}
(Independent progress under crash faults).}
Let $r\notin\mathcal F$ be a non-faulty robot activated at time $t$, with
$r(t)\neq m^\dagger$, and let $t'>t$ be the time at which it completes
the move prescribed by \textsc{MoveToDestination()}. Then
$d(r(t'),m^\dagger)\leq d(r(t),m^\dagger)-1$. The completion of this move
does not require the activation or movement of any other robot.

\begin{proof}
By Lemma~\ref{lem:target-invariant}, every snapshot observed by $r$ yields
the same target $m^\dagger$. The routine
\textsc{MoveToDestination()} prescribes a movement contained in a
Manhattan shortest path toward $m^\dagger$. Since a completed movement
advances $r$ by at least one grid edge along this path,
$d(r(t'),m^\dagger)\leq d(r(t),m^\dagger)-1$.

The prescribed destination is computed from the snapshot available to
$r$ and does not require another robot to move first. A crash-faulty robot
remains fixed at its last position. Since a grid node may contain multiple
robots, such a stationary robot neither blocks the prescribed path nor
prevents $r$ from entering an occupied node. Therefore, $r$ completes its
distance-decreasing movement independently of the activation or movement
of every other robot.
\end{proof}



\noindent\textbf{Lemma~\ref{lem:finite-time-gathering}
(Finite-time gathering).}
Under every fair asynchronous schedule, all non-faulty robots reach
$m^\dagger$ after finitely many completed moves and remain there.

\begin{proof}
Let
$\Phi(t)=\sum_{r\notin\mathcal F}d(r(t),m^\dagger)$, where the sum is over
the robots that never crash. Since the number of robots and all grid
distances are finite, $\Phi(t_0)$ is a finite non-negative integer.

Consider a completed move of a non-faulty robot outside $m^\dagger$.
By Lemma~\ref{lem:crash-noninterference}, the contribution of that robot
to $\Phi$ decreases by at least one. Every other completed movement also
moves its robot toward $m^\dagger$ and therefore cannot increase $\Phi$.
Consequently, after each such completed move from time $t$ to time $t'$,
$\Phi(t')\leq\Phi(t)-1$.

Suppose that $\Phi(t)>0$. Then at least one non-faulty robot is outside
$m^\dagger$. By Lemma~\ref{lem:crash-noninterference}, this robot does not
wait for another robot, and fairness ensures that it is activated again and
completes another prescribed move. Hence, $\Phi$ cannot remain positive
after $\Phi(t_0)$ distance-decreasing completed moves. Therefore, every
non-faulty robot reaches $m^\dagger$ after finitely many completed moves.
Once there, it remains stationary by
\textsc{MoveToDestination()}.
\end{proof}





\end{document}